\documentclass[12pt]{article}
\usepackage[utf8]{inputenc}
\usepackage{tikz}
\usepackage{amsfonts} 
\usepackage{setspace}
\usepackage[plain]{fullpage}
\usepackage{amsmath,amsfonts,amssymb,amsthm,mathtools}
\usepackage[authoryear, longnamesfirst]{natbib}
\usepackage{mathrsfs}  
\usepackage{calligra}
\usepackage{hyperref}
\usepackage{adjustbox}
\usepackage{subcaption}
\usepackage{easyReview}
\theoremstyle{definition}
\newtheorem{defn}{\protect\definitionname}
\providecommand{\definitionname}{Definition}
\usepackage{bm}

\theoremstyle{remark}
\newtheorem{rem}{\protect\remarkname}
\providecommand{\remarkname}{Remark}

\newtheorem{assumption}{Assumption}
\newtheorem{lem}{\protect\lemmaname}
\providecommand{\lemmaname}{Lemma}

\newtheorem{theorem}{\protect\theoremname}
\providecommand{\theoremname}{Theorem}

\title{\textbf{Externally Valid Selection of Experimental Sites via the k-Median Problem}\footnote{We would like to thank Isaiah Andrews, Tim Armstrong, Mich\`ele Belot, Arun Chandrasekhar, Jiafeng (Kevin) Chen, Max Cytrynbaum, Michael Gechter, Jesse Goodman, Kei Hirano, Yiqi Liu, Charles Manski, Francesca Molinari, Guillaume Pouliot, Andres Santos, David Shmoys, Davide Viviano, four anonymous referees at this journal and four more at the $26^{th}$ ACM Conference on Economics and Computation (EC'25), as well as participants at the Bravo/JEA/SNSF Workshop on “Using Data to Make Decisions” (Brown
University), the Yale Research Initiative on Innovation \& Scale 2024 (Y-RISE) conference, and the Econometric Society 2025 World Congress (Seoul) for helpful feedback, comments, and suggestions. We are especially grateful to John List for his insightful comments and suggestions that substantially improved the paper. We would also like to thank Eddie Ramirez Saquic for excellent research assistance. We gratefully acknowledge financial support from the NSF under grant SES-2315600.}}
\author{
\begin{tabular}{c}
    Jos\'e Luis Montiel Olea\thanks{Cornell University. Department of Economics. Email addresses (in alphabetical order): \href{mailto:montiel.olea@gmail.com}{\textit{montiel.olea@gmail.com}}, \href{mailto:bq45@cornell.edu}{\textit{bq45@cornell.edu}}, \href{mailto:cq62@cornell.edu}{\textit{cq62@cornell.edu}}, \href{mailto:stoye@cornell.edu}{\textit{stoye@cornell.edu}}, \href{mailto:ys556@cornell.edu}{\textit{ys556@cornell.edu}}.} \and
    Brenda Prallon\footnotemark[2] \and 
    Chen Qiu\footnotemark[2]
\end{tabular}\\
\begin{tabular}{c}
    J\"org Stoye\footnotemark[2] \and
    Yiwei Sun\footnotemark[2]
\end{tabular}
}

\date{September  2026}

\begin{document}

\maketitle

\onehalfspacing

\begin{abstract}
We present a decision-theoretic justification for viewing the question of how to best choose \emph{where} to experiment in order to optimize external validity as a \emph{$k$-median problem}, a popular problem in computer science and operations research. In particular, when treatment effect heterogeneity across experimental and policy-relevant sites is substantial (in a sense we make precise), we present conditions under which minimizing the \emph{worst-case}, welfare-based regret among \emph{all} nonrandom schemes that select $k$ sites to experiment is equivalent to solving a $k$-median problem. The connection costs in the relevant $k$-median problem are given by ex-ante bounds on worst-case \emph{voltage effects} between sites, and minimizing the sum of worst-case voltage effects can be cast as a linear integer program. Two empirical applications illustrate the theoretical and computational benefits of the suggested procedure.
\end{abstract}

\newpage
\section{Introduction}

A common concern in randomized evaluations of new policies is their \emph{external validity}; that is, whether the estimated effects of a policy intervention carry over to new samples or populations \cite[Chapter 8,][]{duflo2007using,al2012generalizability,list2020non,vivalt2020much, gechter2024generalizing}. \cite{list2022voltage} refers to the fundamental challenge that policymakers face when deciding whether to scale up a policy that has already been experimentally evaluated---but whose effects could plausibly change once it is scaled up---as the \emph{voltage effect} problem.

A recent literature argues that concerns about the external validity of randomized evaluations---and, more generally, the voltage effect problem---can be assuaged by explicitly incorporating this goal into the experimental, for example through careful choice of where to experiment; see \cite{degtiar2023review} for an overview and references, \cite{chassang2022designing} for some recent work, and the \emph{``Option C thinking''} framework of \cite{list2024optimally,list2026economics} that helps researchers put scale-up considerations at the center of research design. For instance, if a researcher has access to multiple sites to experimentally evaluate a new policy, it is possible to use information available before the evaluation (such as site-level characteristics) to \emph{nonrandomly} select one or more sites. Recently, \cite{egamidesigning} and \cite{gechter2023site} argued that a research design which nonrandomly selects experimental sites---referred to broadly in the literature as \emph{purposive sampling} \cite[p. 511]{cook2002experimental}---may improve the external validity of randomized evaluations.

We present a decision-theoretic justification for viewing the question of how to nonrandomly select $k$ sites to experiment---or, equivalently,  how to design a purposive sampling scheme---with the goal of optimizing external validity as a \emph{$k$-median problem}. This is a classical problem in computer science and operations research.\footnote{See Chapter 7.7 in \cite{shmoys2011design} for a textbook treatment of this problem and \citep{cohen2022improved} for a recent discussion of approximate algorithms for it.

In the clustering literature, the $k$-median problem we just described (following the jargon used in computer science and operations research) is often called \emph{$k$-medoids clustering}; for example, see footnote 1 in \cite{ushakov2021near}. These clustering algorithms are also employed in the literature of space-filling experimental designs; see \cite{mak2018minimax} for a review. } Broadly speaking, in a $k$-median problem, there is a set of \emph{facilities} and a set of \emph{clients}; the goal is to open at most $k$ facilities and \emph{connect} each client to at least one facility at minimal total connection cost. The $k$ facilities that solve the $k$-median problem are the locations that are closest to the clients, where closeness is defined by total connection cost. As we will explain below, the $k$ sites selected for experimentation by solving the $k$-median problem will be the $k$ sites that, broadly speaking, \emph{minimize the sum of worst-case voltage effects} (or, as we explain in Section \ref{sec:main}, \emph{cumulative worst-case voltage regret}). We next provide more specific details on how we use statistical decision theory to establish a formal connection between the \emph{site-selection problem} to optimize external validity and the \emph{k-median problem}. 

{\scshape The Site-Selection Problem:} Suppose a policy maker can use some \emph{experimental sites} to gather experimental evidence on the effectiveness of a policy or treatment of interest. Because of logistical or budget constraints, we assume that the policymaker can choose at most $k$ sites. Based on their experimental outcomes, she will decide whether it is worthwhile to scale up the policy or treatment of interest in a set of \emph{policy-relevant sites}. 
How should the policymaker decide where to experiment? \cite{gechter2023site} formalize the site-selection problem using a statistical decision theory framework in which, after experimentation, the policymaker observes a vector of unbiased and normally distributed estimates of the policy effects of interest at the selected experimental sites.\footnote{The baseline model assumes that the policymaker has no access to individual-level covariates. This means that the policymaker does not have enough information to estimate conditional average treatment effects (CATEs) that could be transported (e.g., via reweighting) to learn the treatment effects of interest at policy-relevant sites. At the end of Section \ref{subsection: model site selection}, we discuss how individual-level covariates, if available, could be used in the site-selection problem.}   

{\scshape Connection to the k-median Problem:} Under a condition that we term \emph{bounded and agnostic voltage effects} (along with an additional \emph{coherency} condition on voltage effects that will be explained later), our main result (Theorem \ref{thm:main}) shows that the optimized  \emph{worst-case} (welfare-based) \emph{regret} of \emph{any} purposive sampling scheme that experiments in at most $k$ sites is proportional to the objective function of a $k$-median problem with the following features: the sites available for experimentation are treated as facilities; the sites where the policymaker would like to implement the new policy are treated as clients; the connection cost between clients and facilities is given by the worst-case voltage effect between the corresponding sites. Importantly, the solution of the $k$-median problem is minimax-regret optimal (among purposive sampling schemes) when i) the candidate sites for experimentation and the policy-relevant sites are \emph{disjoint} and ii) the treatment effect heterogeneity across sites accommodated in the parameter space is \emph{substantial} (in a sense Theorem \ref{thm:main} makes precise). We argue that, under the voltage effects framework, assuming that i) holds is without loss of generality. To see this, note that the main role of experimental sites is to provide information about the effects of the policy of interest. Even if one of these sites is also policy relevant (in the sense that the policymaker is interested in scaling up the policy of interest at that specific site), the presence of voltage effects implies a nontrivial extrapolation problem: the effects of a policy after scaling could be substantially different to the effects revealed by the experimental evaluation. Thus, for all practical purposes, any given site can be thought of as being associated with two different treatment effects: the effect at the moment of experimentation, and the effect after the policy of interest has been scaled up. Selecting the $k$ sites that solve the $k$-median problem minimizes the worst-case voltage effects, and hence optimizes external validity (in a minimax-regret sense). 

In order to formalize the connection between the site selection problem and the $k$-median problem, we leverage recent developments in the literature on treatment choice problems with partial identification \citep{yata2021,ishihara2021, olea2023decision}. Note first that in the $k$-median problem, each client is typically connected to only one facility (otherwise, the connection cost would not be minimized). In the context of the site selection problem, a connection between a site $i$ (where an experiment was conducted) and a site $j$ (where no experimentation occurred) means that the estimate from site $i$ informs policy decisions in site $j$. Thus, the solution to the $k$-median problem prescribes to only use the information from the site with the smallest connection cost---the \emph{nearest neighbor}---even when experimental outcomes in $k>1$ sites are available. Why would this apparent discarding of information be decision-theoretically optimal? The reason is that, after $k$ sites have been selected for experimentation, the policymaker faces the ``evidence aggregation'' problem introduced in \cite{ishihara2021} and recently discussed in \cite{yata2021} and \cite{olea2023decision}. This literature has established conditions under which it is minimax-regret optimal to base decisions on the nearest neighbor's data, provided the true treatment effects are allowed to vary substantially; see, for example, \citet[][Proposition 1]{olea2023decision}. The intuition is a simple bias/variance trade-off. Using a convex combination of multiple sites reduces variance, but increases bias. When treatment effect heterogeneity may be large, bias matters more than variance in the worst-case scenario, and bias is minimized by only considering the nearest neighbor. Moreover, in this case, worst-case regret is proportional to the worst-case voltage effect between the site of interest and its nearest neighbor. 

In our site selection problem, rather than simply choosing whether to adopt a program based on available evidence (Option A) or reject it (Option B), the decision maker considers the value of gathering experimental evidence at strategically selected sites before committing to a scale-up decision. Our site selection approach based on the $k$-median problem provides a formal, operational answer to one of the key components of the Option C thinking framework of \cite{list2024optimally,list2026economics}: which sites to select for experimentation. In particular, our suggested purposive sampling scheme can be understood as a strategy---supported by the formal framework of statistical decision theory---for minimizing the worst-case voltage effects attributable to site selection. We note that by choosing the $k$ sites that solve the $k$-median problem, the policymaker is minimizing the sum of worst-case voltage effects between policy-relevant sites and experimental sites, thus bounding the scope for treatment effect attenuation when moving from experiment to implementation.    
{ \scshape Algorithms:} The connection with the $k$-median problem clarifies the problem's difficulty but also suggests efficient algorithms. To see the need for those, recall that any purposive sampling scheme optimizes over ``$n$ choose $k$'' potential site combinations, where $n$ is the total number of sites available for experimentation. 
Thus, optimally choosing a purposive sampling scheme by simply evaluating the performance of each combination is costly when ``$n$ choose $k$'' is large.

Conceptually, the connection with the $k$-median problem allows us to understand the computational complexity of finding an exact minimax-regret optimal purposive sampling scheme under the conditions of Theorem \ref{thm:main}. Since the $k$-median problem is known to be NP-hard \citep{kariv1979algorithmic, megiddo1984complexity,cohen2018bane}, there exists no algorithm that is known to scale polynomially in all of the problem's inputs. 

However, from a practical perspective, it is known that the $k$-median problem admits a linear integer program formulation \citep[Chapter 7.7, p. 185]{shmoys2011design} and is routinely solved using off-the-shelf algorithms such as different versions of the branch-and-bound method \citep[][Chapter 11.1]{bertsimas2005optimization}. In addition, different branch-and-bound algorithms either find a solution with provable optimality or, if stopped early, generate a report on the suboptimality of the solution found \citep{AOSBertsimas}. As we explain later, in our application with $n=41$, any problem for $k \in \{1,\ldots, 10\}$ can be solved to provable optimality in just a few seconds using a personal laptop (see Figure \ref{fig:runtime} in Section \ref{section:application}).    

{\scshape Related Literature:} Our results build on two recent papers that present novel purposive sampling strategies to select experimental sites so as to optimize external validity. \cite{gechter2023site} present an elegant decision-theoretic approach that frames external validity as a policy problem and---under the assumption that the policymaker has a priori information about the effects of the new policy across sites---recommend a Bayesian approach for choosing where to experiment. \cite{egamidesigning} use the principle behind synthetic control \citep{abadie2010synthetic} to recommend the \emph{synthetic purposive sampling} of sites; specifically, they select good \emph{donor} sites whose weighted average of covariates is close to those of the sites of interest.\footnote{The idea of using the synthetic control method for experimental design was first  introduced in \cite{abadie2021synthetic}.} Our approach is related: the site selection achieved by solving the $k$-median problem can be interpreted as degenerate synthetic purposive sampling strategy where each unit's associated synthetic unit is just its nearest neighbor. 

We also contribute to the literature arguing that a ``modern, decision-theoretic framework can help clarify important practical questions of experimental design'' \citep{banerjee2017decision}. Although decision-theoretic approaches to external validity are recent, a large body  of work used statistical decision theory to analyze other aspects of experimental design such as sample size determination \citep{raiffa1961applied,manski2016sufficient,manski2019trial,azevedo2020b,azevedo2023b,hu2024minimax}. Our notion of external validity is conceptually related to areas of research in econometrics, machine learning, and statistics such as domain adaptation \citep{mansour2009domain, ben2010theory}, distributional shifts \citep{duchi2021learning,sugiyama2007covariate,adjaho2022externally}, learning under biased sampling \citep{sahoo2022learning}, and cross-domain transfer estimation and performance \citep{andrews2022transfer, menzel2023transfer}. To the best of our knowledge, none of these papers contains decision-theoretic analyses of \emph{where} to experiment. The site selection problem is also related to optimal regression design (e.g., \citealt{sacks1984some}) and kriging (e.g., \citealt{stein1999interpolation}); both consider a mean square error criterion with either Bayes or minimax optimality, while we focus on minimax welfare regret. See also \cite{karmakar2022approximation} and references therein for the study of blocked randomization designs in stratified experiments to improve precision of treatment effect estimates.

{\scshape Outline:} The remainder of this paper is organized as follows. Section \ref{sec:notation} introduces the formal framework. Section \ref{sec:main} presents our main result linking the $k$-median problem to considerations of external validity. Section \ref{section:kmedian} presents the linear integer program formulation of the $k$-median problem. Section \ref{section:application} presents two illustrative empirical applications. Section \ref{section:extensions} considers extensions of the baseline model. Section \ref{sec:conclusion} concludes. Proofs of the main results can be found in Appendix \ref{sec:app.main}. Additional results are collected in the Supplementary Appendix.

\section{Setting up the Decision Problem} \label{sec:notation}

\subsection{Notation and Assumptions}

A policymaker considers a finite set $\mathcal{S}:=\{1,\ldots,S\}$ of candidate sites to evaluate and, eventually, implement a new policy of interest. For any subset $\tilde{\mathcal{S}}\subseteq\mathcal{S}$, let $\#\tilde{\mathcal{S}}: =\textrm{card}(\tilde{\mathcal{S}})$ denote the number of elements in $\tilde{\mathcal{S}}$. In order to accommodate situations in which the policymaker is not necessarily able to experiment in all the candidate sites, we assume there is a nonempty subset $\mathcal{S}_{E} \subseteq \mathcal{S}$ of what we term \emph{experimental} sites.\footnote{As discussed in \cite{allcott2015site}, there are often systematic reasons determining the eligibility of certain sites for experimentation. For example, in microfinance RCT studies, experiments often require large sample sizes and well-managed microfinance institutions (MFIs), characteristics more commonly found in older and larger institutions. To qualify for clinical trials involving a new surgical procedure, hospitals and surgeons need to have both experience in the procedure and a history of low mortality rates.} Throughout the paper, and to avoid a trivial instance of the site selection problem, we assume that there are at least two experimental sites (i.e., $\#\mathcal{S}_{E} \geq 2$). It is also possible that institutional restrictions preclude the eventual implementation of the policy of interest in all of the candidate sites. Thus, it will be convenient to denote by $\mathcal{S}_{P} \subseteq \mathcal{S}$ the nonempty set of \emph{policy-relevant} sites.  
We allow for \emph{treatment effect heterogeneity} across sites but impose some restrictions on it. In particular, it will be of crucial importance to be explicit about the way in which treatment effects in experimental sites relate to treatment effects in policy-relevant sites. More concretely, denote by $\tau_s$ the treatment effect  for each $s\in\mathcal{S}$, then we impose:

\begin{assumption} \label{asm:voltage effect}

   There exists a known square, symmetric, nonnegative matrix $C$ of dimension $S$ such that for every $s,s'\in \mathcal{S}=\{1,\ldots,S\}$ we have
   \begin{equation} \label{eq:voltage effect} 
    | \tau_{s}-\tau_{s'} | \leq C_{s,s'}. 
    \end{equation}
    
\end{assumption} The main motivation behind Assumption \ref{asm:voltage effect} is that the evidence gathered from experimental sites, while imperfect, is nevertheless informative about the effects of the policy of interest. Importantly, this allows us to capture what \cite{list2022voltage} terms \emph{voltage effects}, a fundamental problem that policymakers face when deciding whether to scale up a policy that has already been experimentally evaluated. \cite{al2023simple} define the voltage effect as \emph{``the propensity for the absolute size of an intervention’s treatment effect to change when the intervention is scaled up (or, more generally, for the benefit-cost profile to change at scale).''} Using this terminology, Assumption \ref{asm:voltage effect} then requires two things. First, the voltage effect associated with sites $s\in \mathcal{S}_{E}, \: s' \in \mathcal{S}_{P}$ is bounded by a known constant $C_{s,s'}$. Second, the policymaker remains agnostic about whether there is a \emph{voltage gain} (in the sense that, after scaling up, the effects at site $s' \in \mathcal{S}_{p}$ can be larger than $\tau_s$) or a \emph{voltage drop} (in the sense that the effects at site $s' \in \mathcal{S}_{p}$ could be smaller than the effects predicted by the experimental evaluation at site $s$). We think that the agnosticism embodied in Assumption \ref{asm:voltage effect} is consistent with the discussion of voltage effects in the literature; for instance, p. 13 in \cite{list2022voltage}: \emph{``Most of us think that scalable ideas have some `silver bullet' feature: some quality that bestows a `can't miss' appeal. That kind of thinking is fundamentally wrong. There is no single quality that distinguishes ideas that have the potential to succeed at scale from those that don't.''} Thus, in words, we could refer to Assumption \ref{asm:voltage effect} as requiring that the treatment effect heterogeneity across sites in our model exhibits \emph{bounded and agnostic} voltage effects.\noindent \footnote{Another important concern is what \cite{list2024optimally,list2026economics} calls a \emph{voltage drop}, i.e., the systematic attenuation of treatment effects when moving from tightly controlled experimental sites to scaled implementation. To account for such voltage drops in our framework, one might generalize \eqref{eq:voltage effect} by breaking its symmetric structure. That is, we may assume $\tau_i+\underline{C}_{i,j}\leq\tau_j\leq\tau_i+\overline{C}_{i,j}$ for any $i\in\mathcal{S}_E$ and $j\in\mathcal{S}_P$. Then, voltage drop may be modeled by allowing one or both of $(\underline{C}_{i,j},\overline{C}_{i,j})$ to be negative. To what extent our main results accommodate such voltage drop is an interesting question for future research. } Finally, Assumption \ref{asm:voltage effect} gives the researcher flexibility to bound the treatment effect heterogeneity between pairs of experimental or policy-relevant sites.

Throughout the rest of the paper,  we assume $\mathcal{S}=\mathcal{S}_{P}\cup \mathcal{S}_{E}$.  We also assume that the experimental and policy-relevant sites are non-overlapping;  that is, $\mathcal{S}_{E} \cap \mathcal{S}_{P} = \emptyset$. This is without mathematical loss of generality because the main role of experimental sites $s \in \mathcal{S}_{E}$ is to provide information about the effects of the policy of interest. Even if one of these sites were also policy relevant (in the sense that the policymaker is interested in scaling up the policy of interest at that specific site), the presence of voltage effects would imply a nontrivial extrapolation problem: the effects of a policy after scaling could be substantially different from the effects revealed by the experimental evaluation. Thus, any given site can be thought as being associated with two different treatment effects: the effect at the moment of experimentation, and the effect after the policy of interest has been scaled up.\footnote{Mathematically, if a particular site can be both experimental and policy relevant, we create two copies of the same site, put one in $\mathcal{S}_E$, and put the other one in $\mathcal{S}_P$. The framework allows for setting voltage effects to $0$ if desired.}

Assumption \ref{asm:voltage effect} does not require (but can accommodate) continuity of treatment effect heterogeneity and is able to capture some forms of unobserved treatment heterogeneity (so that the variation in treatment effects is not explained fully by the observed covariates; see Section \ref{subsec:unobserved}). At this point, it is helpful to provide an example of a simple mathematical model of treatment effect heterogeneity consistent with Assumption \ref{asm:voltage effect}.  

{\scshape Example 1} (Treatment effect heterogeneity based on site-level characteristics): Suppose that for each site $s\in\mathcal{S}$, the policymaker observes a vector of \emph{site-level characteristics} $X_s \in \mathbb{R}^d$. Suppose we assume that treatment effects across sites depend only on these observable characteristics. Formally,  $\tau_s:=\tau(X_s)$, where $\tau:\mathbb{R}^d \rightarrow \mathbb{R}$ reports conditional (on $X$) treatment effects. If we are willing to posit that sites with similar observed characteristics also have similar treatment effects, we could require $\tau(\cdot)$ to be \emph{Lipschitz continuous}: that is, for any $x,x' \in \mathbb{R}^{d}$, $|\tau(x)-\tau(x')| \leq L \| x-x' \|,$ with some known (Lipschitz) constant $L\geq0$, where $\Vert\cdot\Vert$ denotes the Euclidean norm. This model of treatment effect heterogeneity corresponds to setting $C_{s,s'}$ as $L \| X_s-X_{s'} \|$ in our more general framework. In Section \ref{sec:other.resitrction}, we discuss other possible restrictions on the function $\tau(\cdot)$. For example, we discuss a version in which $\tau(\cdot)$ is a \emph{H\"{o}lder continuous} function and also a more general formulation where $\tau(\cdot)$ belongs to a  convex and centrosymmetric space of functions\footnote{Such a restriction has been used recently in the econometrics literature to analyze estimation, inference, and other decision problems that arise in a nonparametric regression setup \citep{yata2021,armstrong2018optimal}.}.   \qed

\subsection{Statistical Model for the Site Selection Problem} \label{subsection: model site selection}

As in \cite{gechter2023site}, the policymaker must choose a strict subset of experimental sites $\mathscr{S} \subset \mathcal{S}_{E}$.\footnote{
We consider the case in which $\mathscr{S}$ is allowed to equal $\mathcal{S}_{E}$ in Section \ref{subsection:fixed_costs}.} Specifically,
there is a known integer $k < \#\mathcal{S}_{E}$ such that $\mathscr{S}$ must belong to the set 
\begin{equation*} \label{eqn:feasible_set} 
\mathcal{A}(k):=  \{ \mathscr{S} \subset \mathcal{S}_{E} \: | \: 1 \leq \#\mathscr{S} \leq k  \}. 
\end{equation*}

\noindent Our notation assumes that at least one site will be chosen for experimentation.  This is without loss of generality as not experimenting at all can lead to arbitrarily large regret.

If the policymaker decides to experiment in a set $\mathscr{S} \in \mathcal{A}(k)$ of cardinality $k(\mathscr{S}):=\#\mathscr{S} \leq k$, she will observe $k(\mathscr{S})$ treatment effect estimates. We collect these estimates in the vector $\widehat{\mathcal{\tau}}_{\mathscr{S}}\in\mathbb{R}^{k(\mathscr{S})}$. We denote the corresponding vector of true treatment effects for the experimental sites in $\mathscr{S}$ as $\tau_{\mathscr{S}}$. We assume that the treatment effect estimators obtained in each site are normally (and independently) distributed around the vector of true effects:
\begin{equation} \label{eqn:stat_model}
\widehat{\tau}_{\mathscr{S}} \sim \mathcal{N}
\left(
\tau_{\mathscr{S}} \: , \: \Sigma_{\mathscr{S}}
\right),
\end{equation}
where $\Sigma_{\mathscr{S}}$ is a diagonal matrix with each entry on the diagonal representing the variance of the corresponding treatment effect estimate.\noindent \footnote{The diagonal structure of $\Sigma_\mathscr{S}$ implies that there is no correlation or other spatial dependence across experimental sites. This requires researchers in the  first-stage experimental design to be   careful in avoiding possible interference or spillovers across different experimental sites. }
Following \cite{gechter2023site}, we furthermore treat $\Sigma_{\mathscr{S}}$ as known.

The normality assumption in \eqref{eqn:stat_model} is unlikely to hold exactly; however, it is common to assume that treatment effect estimates from randomized controlled trials are asymptotically normal with asymptotic variances that can be estimated consistently. Treating the limiting normal model as an exact finite-sample statistical model eases exposition and allows us to focus on the core features of the site selection problem. 
Indeed, working directly with such a limiting model is common in applications of statistical decision theory to econometrics; see  \cite{muller2011efficient} and the references therein for theoretical support and applications in the context of testing problems and \citet{ishihara2021}, \citet{stoye2012minimax}, or \citet{tetenov2012statistical} for precedents in closely related work.

After observing $\widehat{\mathcal{\tau}}_{\mathscr{S}}$, the policymaker chooses an \emph{action} $a_s\in[0,1]$ at each policy-relevant site $s \in \mathcal{S}_{P}$. We interpret this action as the proportion of a site's population that will be randomly assigned to the new policy. Thus, $a_s=1$ means that everyone in site $s$ is exposed to the new policy, and $a_s=0$ means that the status quo at the site is preserved. Under this interpretation, $a_s=.5$ means that 50\% of the population at site $s$ will be exposed at random to the new policy; however, the formal development equally applies to either individual or population-level randomization. Our interpretation abstracts from integer issues arising with small populations. 

Given a selection of sites $\mathscr{S}$, we define a \emph{treatment rule} $T$ as a (measurable) function $T: \mathbb{R}^{k(\mathscr{S})} \rightarrow [0,1]^{\#\mathcal{S}_P}$ that maps experimental outcomes to (possibly) randomized policy actions in each of the policy-relevant sites. It will sometimes be convenient to use $T_{s}(\cdot)$ to denote the specific treatment rule for site $s \in \mathcal{S}_{P}$ implied by $T$. Note that, for notational simplicity, we have not explicitly indexed $T$ and $T_s$ by the selected sites $\mathscr{S}$, although it should be clear from the definitions that these treatment rules indeed depend on $\mathscr{S}$. To remind the reader about this feature, we let $\mathcal{T}_{\mathscr{S}}$ denote the set of all treatment rules given the selection $\mathscr{S}$).\footnote{The set of all possible rules  once the variation in $\mathscr{S}$ is accounted for could then be denoted as $\cup_{\mathscr{S}\in \mathcal{A}(k)}\mathcal{T}_{\mathscr{S}}$.} We call $T\in \mathcal{T}_{\mathscr{S}}$ \emph{nonrandomized} if for every $s \in \mathcal{S}_P$ we have $T_{s}( z ) \in\{0,1\}$ for (Lebesgue) almost every $z \in\mathbb{R}^{k(\mathscr{S})}$. Otherwise, we say that the rule is \emph{randomized}.

For sake of exposition, we initially assume that there is no cost of experimentation. While this assumption is clearly unrealistic, we later show that the main conclusions of our analysis are robust to adding fixed costs to the objective function of the $k$-median problem.

\subsection{Welfare and Regret} 
Let the vector $\tau$ collect all treatment effects $\tau_s$ for $s\in\mathcal{S}$. 
For each policy-relevant site, we assume that its \emph{welfare} under the status quo policy is known and normalized to zero. Furthermore, welfare is linear in the action, i.e., the welfare of each $s\in\mathcal{S}_P$, after receiving treatment rule $T_s$, is $\tau_sT_{s}(\widehat{\tau}_{\mathscr{S}})$.\noindent \footnote{This rules out general equilibrium effects, spillover effects as well as externalities within each site. These assumptions are aligned with most of the existing literature on treatment choice. An exception is  \cite{manski2007admissible}, who consider welfare as a concave transformation of action. } The welfare of a decision rule $T$, given that sites $\mathscr{S}$ are selected for experimentation, corresponds to the average expected welfare across policy-relevant sites:
\begin{equation} \label{eqn:welfare}
\mathcal{W}(T,\mathscr{S},\tau) := \frac{1}{\# \mathcal{S}_{P}}\sum_{s \in \mathcal{S}_{P}} \tau_s \mathbb{E}_{\tau_{\mathscr{S}}}[T_{s}(\widehat{\tau}_{\mathscr{S}})], 
\end{equation}
where $\mathbb{E}_{\tau_{\mathscr{S}}}[T_{s}(\widehat{\tau}_{\mathscr{S}})]$ means that the expectation is taken with respect to  $\widehat{\tau}_{\mathscr{S}}$.

The \emph{regret} of policy $(T,\mathscr{S})$ equals
\begin{equation} \label{eqn:regret}
\mathcal{R}(T,\mathscr{S},\tau) := \frac{1}{\# \mathcal{S}_{P}} \sum_{s \in \mathcal{S}_{P}} \tau_s \bigl( \mathbf{1}\{ \tau_s \geq 0\} - \mathbb{E}_{\tau_{\mathscr{S}}}[T_{s}(\widehat{\tau}_{\mathscr{S}})] \bigr).
\end{equation}
Our focus will be on finding the purposive sampling scheme that minimizes worst-case regret. We think that regret provides a reasonable objective function for the site-selection problem. From the perspective of policymaking, focusing on regret is a way of internalizing the concern for critiques about incorrect \emph{scale-up} decisions. For example, adopting a policy that should have not been adopted, or failing to adopt a policy that could have been beneficial. 

Let $C$ be a matrix of dimension $S\times S$ with $C_{s,s'}$ as its $(s,s')$-th entry. Let $\tau(C)$ denote the set of all treatment effects $\tau$ that are consistent with Assumption \ref{asm:voltage effect} for a given matrix $C$. We will refer to $\tau(C)$ as the \emph{parameter space}.

\begin{defn}[\emph{MMR optimal purposive sampling scheme and treatment rule}]  \label{def:mmr}
The pair $(T^*,\mathscr{S}^*) \in \mathcal{T}_{\mathscr{S}} \times \mathcal{A}(k)$ is minimax-regret (MMR) optimal among all purposive sampling schemes and treatment rules if
\begin{equation} \label{eqn:MMR}
\sup_{\tau \in \tau(C)}  \mathcal{R}(T^*,\mathscr{S}^*,\tau) = \adjustlimits \inf_{\mathscr{S} \in \mathcal{A}(k), T \in \mathcal{T}_{\mathscr{S}}}  \sup_{\tau \in \tau(C)}  \mathcal{R}(T,\mathscr{S},\tau).
\end{equation}
\end{defn}

Minimax regret is popular in the literature on decision theoretic analysis of treatment choice because it is a principled notion that avoids priors while also avoiding a certain ultra-conservativeness of minimax welfare loss that tends to generate trivial solutions in treatment choice problems. See \citet{Manski21,Manski26} for ``agenda papers'' and \cite{Stoye2012new} for a survey and further references. 
Definition \ref{def:mmr} implies an asymmetric treatment of randomization: While we allow the policymaker to randomize policy implementation choices, we are restricting her to pick the experimental sites in a deterministic fashion. In Section \ref{section:extensions}, we discuss challenges we encountered in trying to allow for the random selection of experimental sites.

\begin{rem} \label{rem:double_inf}
It will sometimes be convenient to rewrite the right-hand side of \eqref{eqn:MMR} as
\begin{equation} \label{eqn:MMR-site-selection-scheme}
\inf_{\mathscr{S} \in \mathcal{A}(k)} \left( \adjustlimits \inf_{T \in \mathcal{T}_{\mathscr{S}}} \sup_{\tau \in \tau(C)}  \mathcal{R}(T,\mathscr{S},\tau) \right). 
\end{equation}
This suggests that, conceptually, the MMR problem can be solved in two steps. First, analyze the problem of policy implementation given the experimental outcomes at sites $\mathscr{S}$; then, optimize over the sites where to experiment. 

This distinction is helpful because the first step is related to \citeauthor{ishihara2021}'s  (\citeyear{ishihara2021}) ``evidence aggregation'' problem, a key difference being that we typically have more than one policy-relevant site.
\qed
\end{rem}

\section{Main Result} \label{sec:main}
This section presents our main result: Finding the minimax-regret optimal purposive sampling scheme is, under some assumptions, exactly equal to solving a $k$-median problem.
In order to provide a formal description of this result, we introduce some additional notation. For each site $s\in\mathcal{S}$ and every  $\mathscr{S}\in\mathcal{A}(k)$, denote by $N_{\mathscr{S}}(s)\in\mathscr{S}$ its (worst-case) \emph{nearest neighbor} in $\mathscr{S}$   (or the nearest neighbor with the smallest index in case of multiplicity).  That is, for every $s \in \mathcal{S}$:
 \begin{equation} \label{eqn:nearest-neighbor-definition}
 C_{N_{\mathscr{S}}(s),s} \leq C_{s',s} \; \: \forall s' \in \mathscr{S}.
 \end{equation}
What we have in mind is that, under Assumption \ref{asm:voltage effect}, the largest possible discrepancy between $\tau_s$ and the treatment effect of any other experimental site $s' \in \mathscr{S}$ is $C_{s',s}$. Therefore, in a worst-case scenario, the site $N_{\mathscr{S}(s)}$ is the closest to $s$. The k-median problem can be defined using the nearest-neighbor assignment as follows. 

\begin{defn}[\emph{$k$-median problem}] \label{def:k-median}
We say that a purposive sampling scheme, $\mathscr{S} \in \mathcal{A}(k)$, solves the $k$-median problem---given a square, symmetric, nonnegative, matrix $C$---if it solves
\begin{equation} \label{eqn:k-median}
\inf_{\mathscr{S}\in\mathcal{A}(k)}\sum_{s\in\mathcal{S}_{P}}C_{N_{\mathscr{S}}(s),s}.
\end{equation}
\end{defn}
See \cite[][Chapter 7.7]{shmoys2011design} for a common definition of the $k$-median problem used in the computer science and operations research literature and \cite[][p. 345]{ushakov2021near} for a common definition used in the clustering literature. In the introduction, we presented a verbal description of the $k$-median problem in terms of \emph{clients, facilities, and connection costs}. We now explain how these terms relate to the site selection problem. In Definition \ref{def:k-median}, the policy-relevant sites (with indices in $\mathcal{S}_{P}$) can be viewed as \emph{clients} and the experimental sites (with indices in $\mathcal{S}_{E}$) can be viewed as \emph{facilities}. The matrix $C$---containing the bounds on voltage effects---can now be viewed as containing connection costs between clients and facilities. The goal is to choose the $k$ facilities that minimize connection cost, thus every client $s \in \mathcal{S}_{P}$ will be connected to the facility that is closest among those in $\mathscr{S}$. When the matrix $C$ is interpreted as collecting connection costs (worst-case voltage effects), the facility in $\mathscr{S}$ that minimizes the connection to a client $s$ is $N_{\mathscr{S}}(s)$. 

We next provide a verbal description of why the bounds on voltage effects are associated with the policymaker's \emph{regret} after scaling up a policy of interest. Recall that, for experimental site $j$ and policy-relevant site $s$, $C_{j,s}$ is the largest (worst-case) voltage effect that a policy maker can experience when scaling up the policy of interest at site $s$ after experimenting on site $j$. If the policy maker scales up the policy evaluated at site $s$, hoping to observe an effect $\tau_j$ in site $s$, in the worst-case scenario she will experience a voltage drop of $C_{j,s}$.  Conversely, if she decides not to scale up the policy, in the worst-case scenario the true effect could correspond to a voltage gain of $C_{j,s}$. In both cases, we can interpret the bound $C_{j,s}$ as \emph{voltage regret} that could be experienced by the policy maker in a worst-case scenario. Thus, the site $N_{\mathscr{S}}(s)$ can be interpreted as the site that minimizes (worst-case) voltage regret for the policy-relevant site $s$. The experimental sites that solve the $k$-median problem then minimize the \emph{cumulative (worst-case) voltage regret.}

Theorem \ref{thm:main} below formalizes this connection between minimax-regret selection of experimental sites and the $k$-median problem. Doing so requires additional, mild regularity conditions on the matrix $C$. In particular, we assume:
\begin{assumption} \label{asm:positive.c}
For all $j,j',s\in\mathcal{S}$, $C_{j,j'}\leq C_{j,s}+C_{j',s}$.
\end{assumption}

\noindent Assumption \ref{asm:positive.c} imposes some minimal coherency on the bounds on voltage effects across sites. For instance, let $s$ be a policy-relevant site. Suppose that the voltage effects $C_{j,s}$ and $C_{j',s}$ associated with scaling up the policy of interest at site $s$ after experimental evaluation at experimental sites $j$ and $j'$ are both close to zero. Then it seems reasonable to say that the treatment effect heterogeneity between sites $j$ and $j'$ cannot be too large. Mathematically, we enforce this restriction by requiring that the square, symmetric, nonnegative matrix $C$ respects the triangle inequality.

Let $\sigma_{N_{\mathscr{S}}(s)}>0$ denote the standard deviation of the treatment effect estimator obtained from site $N_{\mathscr{S}}(s)\in\mathscr{S}$. Define 

\begin{equation}\label{eqn:choice.C.star}
C^*:= \underset{\mathscr{S}\in \mathcal{A}(k)}{\max}C(\mathscr{S})<\infty, \quad \textrm{where} \quad C(\mathscr{S}):=\sqrt{\pi/2}\max_{s\in\mathcal{S}_{P}}\left(\sigma_{N_{\mathscr{S}}(s)}\right).    
\end{equation}

\begin{theorem} \label{thm:main}
Suppose Assumptions \ref{asm:voltage effect}-\ref{asm:positive.c} hold. If all off-diagonal elements of $C$ are strictly larger than $C^*$, then: 
\begin{equation}
\label{eqn:main_thm}
 \inf_{T \in \mathcal{T}_{\mathscr{S}}} \sup_{\tau \in \tau(C)}  \mathcal{R}(T,\mathscr{S},\tau) =  
\frac{1}{2}\frac{1}{\#\mathcal{S}_{P}}\sum_{s\in\mathcal{S}_{P}}C_{N_{\mathscr{S}}(s),s}, \quad \textrm{for all }\mathscr{S} \in \mathcal{A}(k).
\end{equation}
Consequently, the minimax-regret purposive sampling scheme for the site selection problem in \eqref{eqn:MMR-site-selection-scheme} can be obtained by solving the $k$-median problem in \eqref{eqn:k-median}.
\end{theorem}  
\begin{proof} 
See Appendix \ref{section:proof_Theorem1}. 
\end{proof}

To account for presence of voltage effects, we work with a general setup in which experimental and policy-relevant sites are non-overlapping (i.e., $\mathcal{S}_{E} \cap \mathcal{S}_{P} = \emptyset$). The left-hand expression in \eqref{eqn:main_thm}, i.e., $\inf_{T \in \mathcal{T}_{\mathscr{S}}} \sup_{\tau \in \tau(C)}  \mathcal{R}(T,\mathscr{S},\tau)$, is the \emph{worst-case regret} of the purposive sampling scheme $\mathscr{S}$ assuming optimality of subsequent treatment choices, while the right-hand side of \eqref{eqn:main_thm} is  the objective function of the \emph{k-median} problem in Definition \ref{def:k-median} (after scaling by the factor $1/(2 \# \mathcal{S}_{P}$)). Thus, Theorem \ref{thm:main} shows that whenever the \emph{worst-case} voltage effects are sufficiently large (in the sense that all off-diagonal elements of $C$ are strictly larger than $C^*$), any solution of the $k$-median problem is an \emph{exact} minimax-regret solution for the site selection problem. This follows from the fact that if a purposive sampling scheme, $\mathscr{S}^*$, minimizes the right-hand side of Equation \eqref{eqn:main_thm}, then it will also minimize the left-hand side, which, by Remark \ref{rem:double_inf}, defines a minimax-regret optimal purposive sampling scheme. We also note that the choice of $C^*$ in \eqref{eqn:choice.C.star} provides  a sufficient condition for this to hold. Given a specific problem at hand,  the matrix $C$ often has more structure, implying that one may be able to find a tighter threshold for which \eqref{eqn:main_thm} is true. We also note that Theorem \ref{thm:main} does not cover the conventional overlapping-site model without an the additional assumption that within-site voltage effects are large enough.  
Theorem \ref{thm:main} follows from bounding  the site selection problem's MMR value from above and below. As shown in Appendix \ref{subsection:proof_Lemma1}, the upper bound arises from bounding the worst-case sum of regrets across sites by the corresponding sum of worst-case regrets for each policy-relevant site. This site-by-site worst-case regret is obtained by applying results from \citet{yata2021} and \citet{olea2023decision} that solve \citeauthor{ishihara2021}'s (\citeyear{ishihara2021}) treatment choice problem for a single site and multiple experimental estimates. In particular, the optimal treatment rule for each site $s \in \mathcal{S}_{P}$,  when $C>C^*$,  depends only on the estimate from its nearest neighbor and may randomize; see Appendix \ref{subsection:proof_Lemma1} for its exact form. While the upper bound can be derived by ignoring cross-site restrictions and applying existing results in the literature, providing a valid lower bound that matches the upper bound is nontrivial and requires new results because we cannot ignore the cross-site restrictions implicitly imposed by Assumption \ref{asm:voltage effect}. In Appendix \ref{subsection:proof_Lemma2}, we show that the upper bound is indeed a valid lower bound under the additional Assumption \ref{asm:positive.c} by a symmetrization argument which may be of independent interest.

\begin{rem}\label{rem:weighted.welfare}

Our results extend to non-equal weights on policy relevant sites with minor modifications.
Denote by $\omega_{s}>0$ a known weight for each site $s\in\mathcal{S}_{P}$.
Without loss of generality, assume that $\sum_{s\in\mathcal{S}_{P}}\omega_{s}=\#\mathcal{S}_{P}$.
We may define weighted welfare as
\[
\mathcal{W}^{\omega}(T,\mathscr{S},\tau):=\frac{1}{\#\mathcal{S}_{P}}\sum_{s\in\mathcal{S}_{P}}\omega_{s}\tau_s\mathbb{E}_{\tau_{\mathscr{S}}}[T_{s}(\widehat{\tau}_{\mathscr{S}})],
\]
with weighted regret
\begin{equation}
\mathcal{R}^{\omega}(T,\mathscr{S},\tau):=\frac{1}{\#\mathcal{S}_{P}}\sum_{s\in\mathcal{S}_{P}}\omega_{s}\tau_s\left(\mathbf{1}\{\tau_s\geq0\}-\mathbb{E}_{\tau_{\mathscr{S}}}[T_{s}(\widehat{\tau}_{\mathscr{S}})]\right).\label{eqn:regret}
\end{equation}
Then, the minimax-regret optimal purposive sampling scheme and treatment rule
are as in Definition \ref{def:mmr} but with the modified $\mathcal{R}^{\omega}$.
Inspecting the proof of Theorem \ref{thm:main}, we find that solving the modified problem is equivalent to solving the alternative
$k$-median problem 
\[
\inf_{\mathscr{S}\in\mathcal{A}(k)}\sum_{s\in\mathcal{S}_{P}}\omega_{s} C_{N_{\mathscr{S}}(s),s},
\]
i.e., the connection cost between client $j$ and facility $i$ is
now $\omega_{j}C_{i,j}$.
\qed\end{rem}

We finish this section by mentioning that we have not been able to find an exact MMR sampling scheme---whether purposive or randomized---for \emph{all} possible values of the bounds on voltage effects $C$. One challenge is that, absent enough treatment effect heterogeneity, the optimal treatment rule for each policy-relevant site will tend to aggregate information from multiple experimental sites and will generally trade off representativeness of a candidate set of experimental sites against the precision of the corresponding treatment effect estimators. To see this, consider the extreme case where there is no site-level heterogeneity; i.e., all entries of $C$ equal zero. In this scenario, any experimental site provides an unbiased estimator for the true
treatment effect in any policy-relevant site.  Consequently, for each set of experimental sites $\mathscr{S}$ and policy-relevant site $s\in\mathcal{S}_{P}$, the optimal treatment rule is a weighted linear combination of the estimated treatment effects in the experimental sites:
\[ T_s(\widehat{\tau}_{\mathscr{S}})=\mathbf{1}\left\{ w_{\mathscr{S}}^{\top}\widehat{\tau}_{\mathscr{S}}\geq0\right\},
\textrm{ where } w_{\mathscr{S}}=\bigl(w_{\mathscr{S}_{1}},w_{\mathscr{S}_{2}},\ldots, w_{\mathscr{S}_{\#\mathscr{S}}}\bigr)^{\top}
\text{ and } w_{\mathscr{S}_{i}}=\frac{1/\sigma_{\mathscr{S}_{i}}^{2}}{\sum_{i=1}^{\#\mathscr{S}}\left(1/\sigma_{\mathscr{S}_{i}}^{2}\right)}.\]
Thus, the optimal sampling scheme would simply choose the $k$
sites in $\mathcal{S}_{E}$ that lead to the smallest-variance estimator of the true treatment effect. 

In the more general case, it is possible to characterize the optimal treatment rule at each policy-relevant site \citep{olea2023decision}. Implementing this site-by-site optimal treatment rule (which provides an obvious upper bound to the problem's worst-case regret) has the additional challenge that the variances
of hypothetical estimates from the experimental sites must be known before the experimentation stage. In Appendix \ref{sec:vol.not.large}, we discuss how to implement such an upper bound using a mixed non-linear integer program for cases where this holds. 

A strength of the $k$-median approach is that---for implementation of the site-selection scheme---it does not require ex-ante knowledge of these variances.\footnote{The implementation of the treatment rule requires knowledge of the variances at the sites that were selected for experimentation, but this will be available after the experimentation stage.} In addition, for general values of $C$, it can also be viewed as an  upper bound on the problem's MMR value: Just i) force the treatment decision $T_s$ to only depend on site $N_{\mathscr{S}}(s)$ and ii) compute worst-case expected regret separately across sites, ignoring that the regret-maximizing parameter configurations may be mutually inconsistent across sites. Both manipulations increase regret and therefore define an upper bound, which is furthermore easy to compute for any given sampling scheme because the induced MMR treatment choice problem was solved in \citet{stoye2012minimax}. For $C$ large enough, the bound will recover Theorem \ref{thm:main}. 

In summary, while the implementation of the $k$-median site-selection scheme does not require knowledge of the variances of all sites, the certification given by Theorem \ref{thm:main} does make use of this information. This means that in applications, obtaining the $k$-median selection scheme is possible even when the matrix of voltage effects is specified up to a constant and the variances of the estimated treatment effects in the experimental sites are unknown (this is the set up in Section \ref{section:application}).\footnote{Moreover, the $k$-median solver only needs the specification of the voltage effects between experimental and policy-relevant sites. Formally satisfying Assumptions \ref{asm:voltage effect} and \ref{asm:positive.c} would simply require assuming a full symmetric completion of the matrix using the same covariate distance metric.} In the experimentation stage that follows after site selection, the researcher will typically have access to the variances of treatment effects at the sites selected for experimentation. These variances are then useful for choosing (or validating) a given proportionality constant; for example, by evaluating $C(\mathscr{S})$ at the $k$-median site selection. Evidently, this information is still not enough to verify the conditions of Theorem \ref{thm:main}, which involve the variances in the experimental sites that were not selected. At that point, we think that verifying the conditions of Theorem \ref{thm:main} requires a priori (and typically unverifiable) assumptions. Some of this a priori information can come from previous experiments or experiments at pilot sites. For instance, if there is information available about the variance of the estimated treatment effect at a pilot experiment, one can make extrapolating assumptions to suggest values of $\sigma_s$ at all the available experimental sites.

\section{Integer Programming and the $k$-Median Problem}
\label{section:kmedian}
The $k$-median problem in \eqref{eqn:k-median} can be formulated as the following linear integer program \citep[Chapter 7.7, p. 185]{shmoys2011design}:
\begin{align*}
    \min_{\{y_i,x_{i,j}\}_{i \in \mathcal{S}_{E}, j \in \mathcal{S}_{P}}} &\sum_{i \in \mathcal{S}_E, j \in \mathcal{S}_P } x_{i,j}\cdot c(i,j)\\
\text{such that } &\sum_{i \in \mathcal{S}_E} x_{i,j} = 1, \;\;\;\;\; \forall j \in \mathcal{S}_P,  \\
    &\sum_{i \in \mathcal{S}_E} y_i \le k,\\
    & 0 \le x_{i,j} \le y_i, \;\;\;\;\; i \in \mathcal{S}_E, j \in \mathcal{S}_P,\\
    &  y_i \in \{0,1\}, x_{i,j} \in \{0,1\}, \;\;\;\;\; i \in \mathcal{S}_E, j \in \mathcal{S}_P. \end{align*}
Here, the choice variables are the binary-valued $y_i$ and $x_{i,j}$, where $y_i$ indicates whether site $i$ is chosen for experimentation and $x_{i,j}$ indicates whether site $i$ is used for policy choices at the policy-relevant site $j$. 
The total number of sites chosen for experimentation cannot exceed $k$, and site $i$ can only be used for making policy choices at policy site $j$ if site $i$ is indeed chosen for experimentation. The connection cost between facility $i$ and client $j$ is $c(i,j)$.

A major advantage of this integer programming formulation is that most ecosystems for scientific computing offer different solvers for linear and nonlinear integer programs. For the applications in this paper, we use the \texttt{MIP} solver in \texttt{Gurobi} \citep{gurobi} through the Python package \texttt{gurobipy}. The \texttt{Gurobi} software is highly optimized, especially for large-scale problems, and it offers an academic license. Even though the scale of the applications presented in Section \ref{section:application} is not large enough for the efficiency advantages of \texttt{Gurobi} to become salient over other solvers, we wanted to showcase its ease of use. It also integrates seamlessly with Google Colab, providing us a way to build self-contained and reproducible examples.\footnote{\href{https://colab.research.google.com/}{https://colab.research.google.com/}}

The \texttt{MIP} solver uses a version of the branch-and-bound algorithm;  see  \cite{bertsimas2005optimization}, Chapter 11.1, for a general description of this algorithm. Broadly speaking, this algorithm first finds the solution to the linear programming (LP) relaxation of the original integer problem. This is known as the \emph{relaxation} step. The problem is then split into two sub-problems (\emph{branching}) that impose the integer constraints on any of the variables that did not have an integer solution. These two steps give bounds on the solution of the program of interest; the algorithm is applied recursively until the lower bound and the upper bound converge up to a tolerance parameter. The recursion creates \emph{nodes}, and there are some strategies to determine which nodes should be explored further; for example, nodes that have integer solutions do not typically require any more branching.

\texttt{Gurobi} implements several additional steps that help the branch-and-bound algorithm be more efficient.\footnote{See \href{https://www.gurobi.com/resources/mixed-integer-programming-mip-a-primer-on-the-basics/}{https://www.gurobi.com/resources/mixed-integer-programming-mip-a-primer-on-the-basics/}.} The \emph{presolve} step reduces the number of effective constraints of the problem by checking if the integer requirement can eliminate some of them. As the name suggests, it is performed before the start of the branch-and-bound algorithm. \emph{Cutting planes} tightens the feasible region by adding linear inequalities to eliminate fractional solutions; it is performed during the branch-and-bound process, and for this reason, the algorithm used by \texttt{Gurobi} can be referred to as a version of a ``branch-and-cut" algorithm. Finally, the \texttt{MIP} solver implements several \emph{heuristics}, for example by rounding the component of a solution that is closest to an integer, fixing it, and hoping the other components will converge to integers more quickly.

In principle, one can always solve \eqref{eqn:k-median} by enumerating all possible size $k$ subsets of $\mathcal{S}_{E}$. Once the matrix of connection costs is supplied, such an algorithm runs in time proportional to ${{\#\mathcal{S}_{E}}\choose{k}}    \cdot \#\mathcal{S}_{P} \cdot k$ and therefore scales poorly when ${{\#\mathcal{S}_{E}}\choose{k}}$ is large. However, we are able to evaluate the performance of our preferred solver by brute-force solving smaller but nontrivial instances of the problem, notably the entire application in Section \ref{sec:multicountry}. 

Figure \ref{fig:gurobi_output} in Appendix \ref{sec:Gurobi} shows an example of the output obtained after using the \texttt{MIP} solver in \texttt{Gurobi} to solve the linear integer program for the application in Section \ref{sec:multicountry}. In this example, the scale of the problem is given by $\#\mathcal{S}_E = \#\mathcal{S}_P = 15, k = 6$, and $d = 8$. We defer the details of the application to Section \ref{sec:multicountry}. 

While we do not use it in this paper, we finally note that there is a large literature studying efficient (polynomial-time) \emph{approximation} algorithms for the $k$-median problem, going back to seminal work of \cite{charikar2002constant}. A basic idea in these algorithms is to consider a linear programming relaxation of the integer program associated with the $k$-median problem \citep[Chapter 7.7]{shmoys2011design}. Even though the scale of the problems analyzed in this paper does not require the implementation of any of these algorithms, there are several papers that present theoretical performance guarantees for them; see for example \cite{cohen2022improved} and also the references in \cite{cohen2018bane}. Finally, it is worth mentioning that when $k$ is fixed (and not viewed as part of the problem's input), there is an approximate algorithm that runs in time proportional to $n$; see  \cite{kumar2010linear}. Such an algorithm could be useful when $n$ is large and $k$ is small.

\section{Applications} 
\label{section:application}

\subsection{Mobile Financial Services in Bangladesh} \label{subsection:bangla}

\cite{lee2021poverty} conducted a randomized controlled trial in Bangladesh to estimate the effects of encouraging rural households to receive money transfers from migrant family members. They specifically conducted an encouragement design where poor rural households with family members who had migrated to a larger urban destination received a 30--45 minute training about how to register and use the mobile banking service ``bKash'' to send instant remittances back home.

The experiment was conducted in the Gaibandha district, one of Bangladesh's poorest regions. It focused on households that had migrant workers in the Dhaka district, the administrative unit in which the capital of Bangladesh is located. \cite{lee2021poverty} measured several outcomes of both receiving households and sender migrants; see their Figures 3 and 4. To give a concrete example of the measured outcomes, one question of interest is whether families that adopt the mobile banking technology are more (or less) likely to declare that the \emph{monga}, the seasonal period of hunger in September through November, was not a problem for their household. \cite{lee2021poverty} (Table 9, Column 7) present results for this specific variable showing that households that used a bKash account in the treatment group are 9.2 percentage points more likely to declare that \emph{monga} was not a problem.  

We ask: Do the findings in \cite{lee2021poverty} generalize to other migration corridors, i.e., combinations of origin and destination districts, in Bangladesh? Is the corridor selected by \cite{lee2021poverty} a good choice for a researcher who is concerned about external validity? Following \cite{gechter2023site}, we name the corridors using a destination-origin format; for example, the migration corridor studied in \cite{lee2021poverty} is ``Dhaka-Gaibandha''.\footnote{This corridor was originally selected since it had qualities that suggested that the intervention might have real impact (as we learned from personal communication with the authors). Also see additional discussions on this in \cite{gechter2023site}.  } Figure \ref{fig:corridors} displays this corridor along with other common ones. The $41$ migration corridors analyzed in \cite{gechter2023site} are depicted as dotted lines connecting an origin and a destination.\footnote{We thank Michael Gechter for gracefully sharing part of the dataset used in \cite{gechter2023site}.}

\begin{figure}[!ht]
    \begin{center}
        \includegraphics[width=\linewidth]{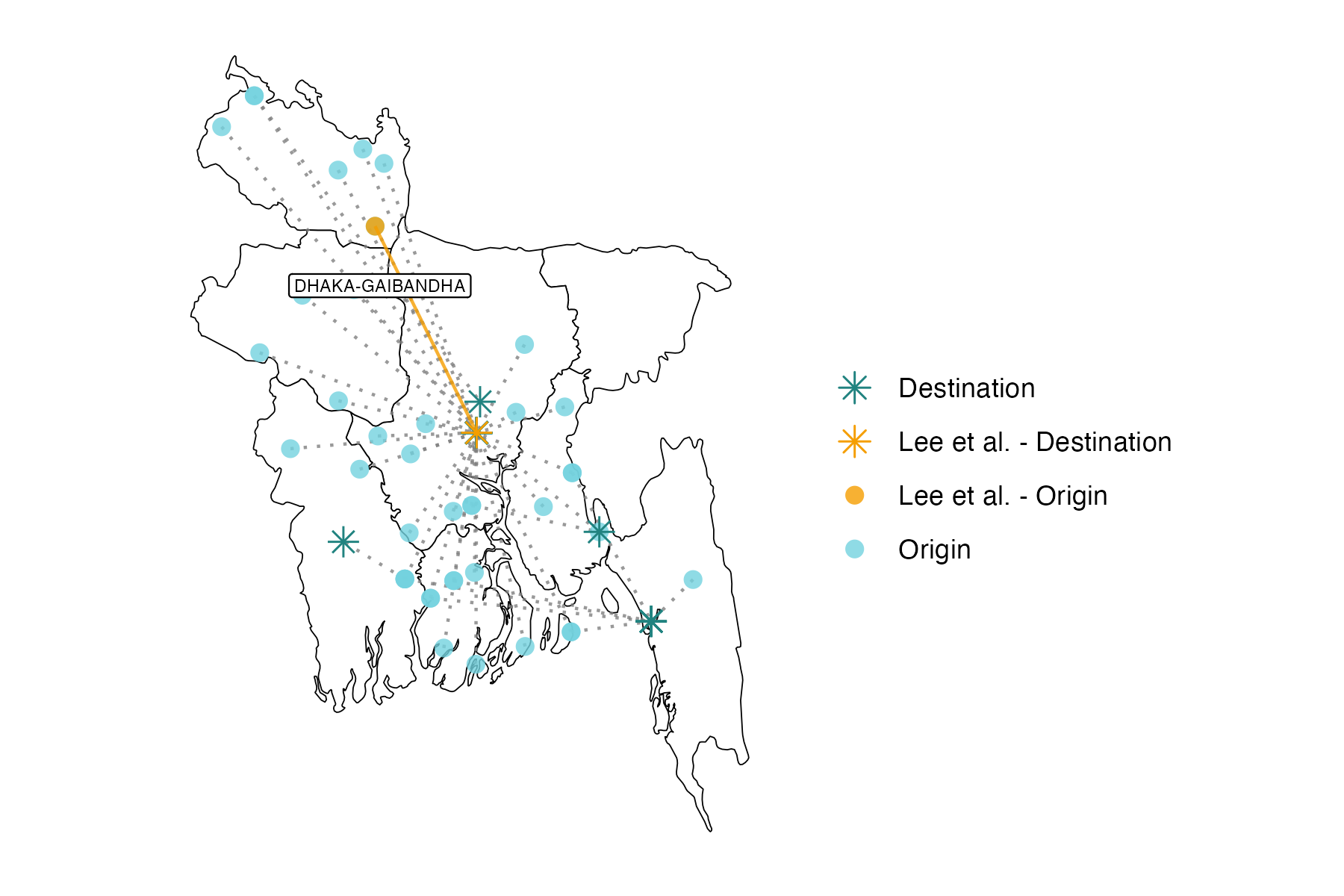}
    \caption{Bangladesh Migration Corridors}
    \label{fig:corridors}
    \end{center}
    {\raggedright \footnotesize \textit{Notes}: The map of Bangladesh with the origins of the migration corridors marked as light blue dots and the destinations marked as dark blue stars. Following the terminology used in \cite{gechter2023site}, \emph{origin} refers to a worker's home, and destination refers to where the worker migrates for work. The corridor where the experiment was originally implemented in \cite{lee2021poverty}, Dhaka-Gaibandha, is highlighted in yellow.}
\end{figure}

In \citeauthor{lee2021poverty}'s \citeyear{lee2021poverty} words, ``[t]he particular nature of our sample potentially limits the external validity'' of the analysis. In short, migration corridors differ in characteristics ranging from distance between origin and destination to cost of living and average wages. Figure \ref{fig:characteristics_gechter} presents a box plot of $d= 13$ (standardized) characteristics collected in a vector $X_i$ that \cite{gechter2023site} used of each of the $i \in \{1,\ldots, 41\}$ migration corridors. We take these corridors to be both our potential experimental and policy-relevant sites; thus, $\#\mathcal{S}_E=\#\mathcal{S}_P=41$. We set $\mathcal{S}=\{1,\ldots, 82\}$, where the first $41$ elements index experimental sites and the next $41$ entries index the same sites, in the same order, but interpreted as policy sites. For any $i\in\mathcal{S}$, we set $C_{i,i}=0$. For any $i \in \mathcal{S}_{E}$ and $j \in \mathcal{S}_{P}$ with $j \neq (41+i)$, we set $C_{i,j}\propto\| X_i - X_{j-41} \|$. If $j=41+i$, we set $C_{i,41+i}\propto\min_{i' \in \{1,\ldots, 41\} \backslash i }\| X_{i'} - X_i \|(1-\epsilon)$ with some $0\leq\epsilon<1$. This guarantees that the largest voltage effect that we anticipate for a policy that is experimentally evaluated at site $i$ and scaled up in the same site is smaller than the largest voltage effect potentially arising from scaling up the policy of interest at any other site.  Below we present results for $k \in \{1,2\}$ and $\epsilon=0.01$.

\begin{figure}[h!]
	\begin{center}
	    \begin{subfigure}{.8\linewidth}
		\centering
		\includegraphics[width=\linewidth]{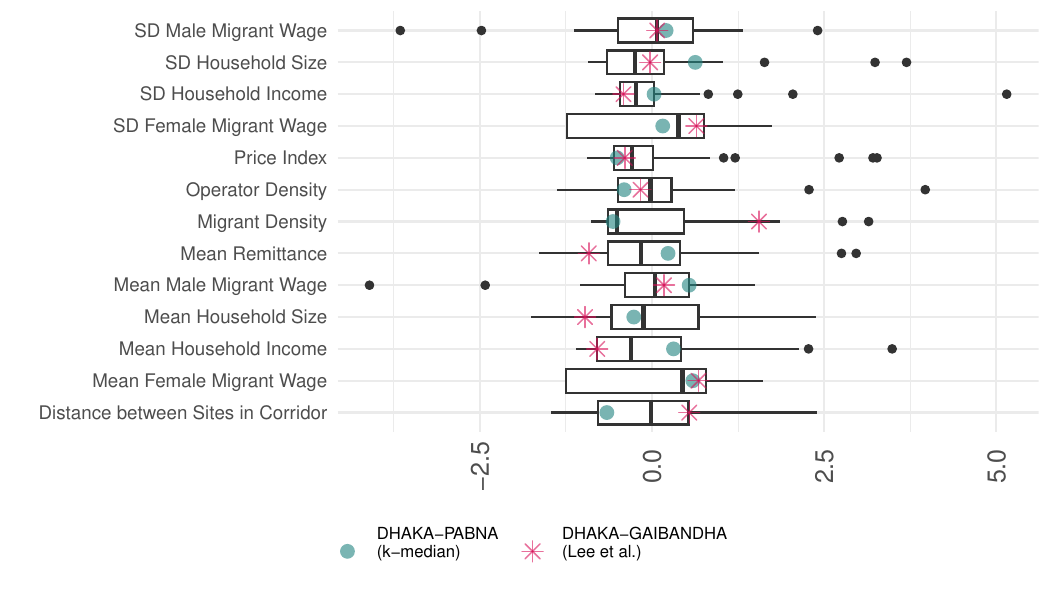}
		\caption{$k$-median}
		\label{subfig:covariates_k_lee_1}
	\end{subfigure}
	\begin{subfigure}{.8\linewidth}
		\centering
		\includegraphics[width=\linewidth]{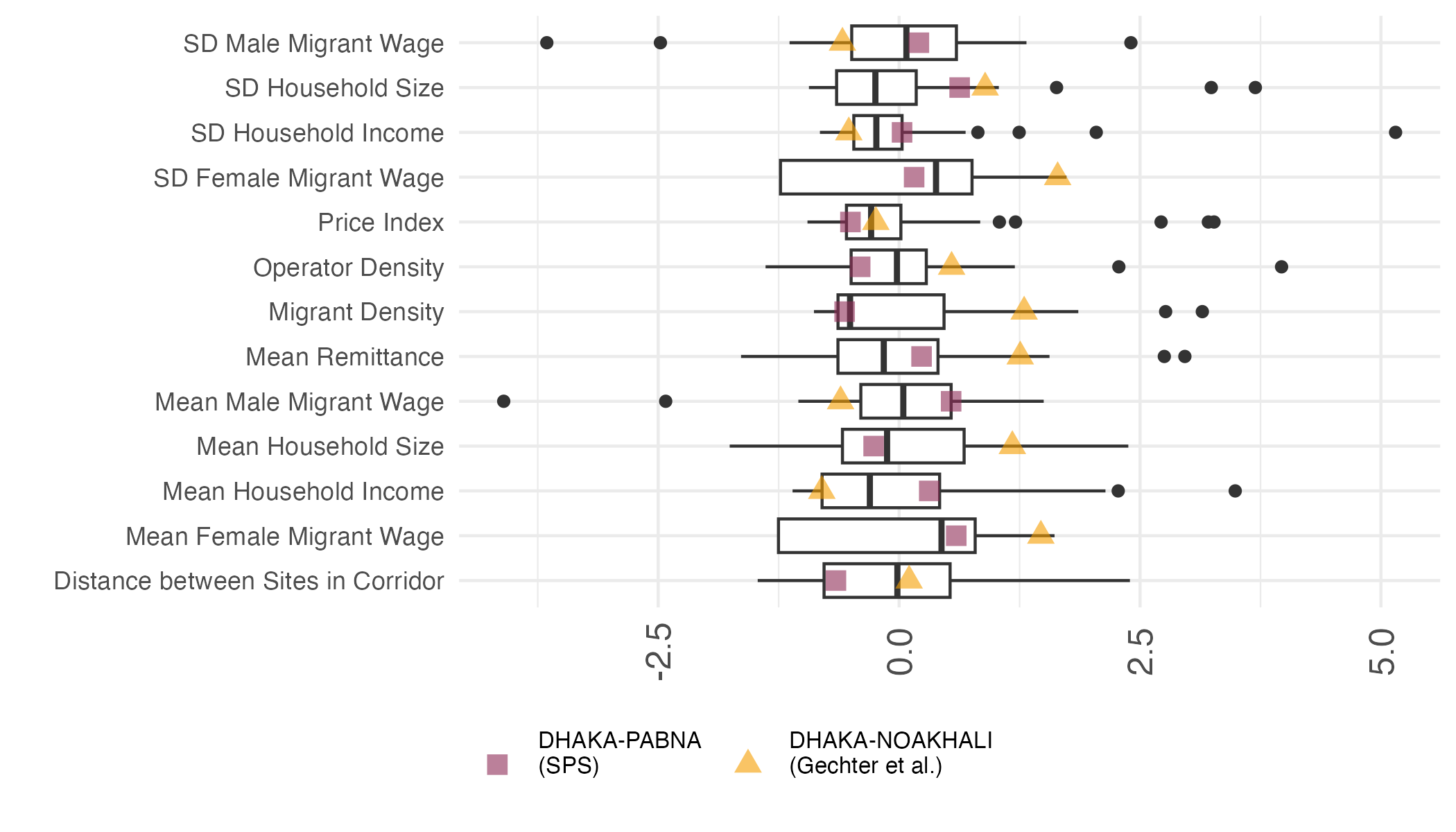}
		\caption{\cite{egamidesigning} and \cite{gechter2023site}}
		\label{subfig:covariates_sps_gechter_1}
	\end{subfigure}
	 \caption{Solution of the $k$-median Problem ($k=1$) and Distribution of Site-Level Covariates}
	 \label{fig:characteristics_gechter}
	\end{center}
 {\raggedright \footnotesize \textit{Notes}: For each covariate, the box represents the interquartile range (IQR), the vertical black line represents the median, and the horizontal line shows the ``theoretical minimum'' (defined as $Q_1 - 1.5$IQR) and ``theoretical maximum'' (defined as $Q_3 + 1.5$IQR). Black dots are outliers, defined as observations that fall beyond the theoretical minimum and maximum. Each panel depicts the sites selected by the different approaches when $k=1$.}
\end{figure}

{\scshape Selected Site when $k=1$:} When $k=1$, the selected site based on the $k$-median problem is Dhaka-Pabna. This is also the solution obtained by using the synthetic purposive sampling approach (henceforth, SPS) in \cite{egamidesigning}. 

Figure \ref{subfig:covariates_k_lee_1} presents a visual comparison of the covariates of Dhaka-Pabna (blue circles) and Dhaka-Gaibandha (pink stars), the original site selected by \cite{lee2021poverty}. The covariate values for Dhaka-Gaibandha are slightly outside the interquartile range for at least three covariates: migrant density, mean remittance, and mean household size. In comparison, all but one covariate value for Dhaka-Pabna are within the interquartile range. The figure also shows that two key covariates of Dhaka-Gaibandha are right at the edges of the interquartile range: distance between sites in the corridor (3rd quartile) and mean household income (1st quartile). One might conjecture that the effects of adopting a mobile banking technology to transfer money particularly depend on distance and household incomes, suggesting that the Dhaka-Gaibandha corridor may not be the most representative.\footnote{\cite{gechter2023site} suggest that these qualities could explain the large treatment effects found by \cite{lee2021poverty}.} Dhaka-Pabna has opposite features: The distance between destination and origin in this corridor is short (1st quartile) and households are relatively better off in terms of income (3rd quartile). The use of the minimax criterion might explain why a corridor with these characteristics may be a good choice for extrapolating experimental results.       

Figure \ref{subfig:covariates_sps_gechter_1} also presents the covariates of Dhaka-Noakhali (yellow triangles), the migration corridor selected by the Bayesian approach of \cite{gechter2023site}.\footnote{By construction, the solution of \cite{gechter2023site} depends on the choice of prior. The results herein reported are based on their preferred prior specification; see Section 5.3 p.p.23 in \cite{gechter2023site}.} For 10 out of the 13 variables that control treatment effect heterogeneity, Dhaka-Noakhali has covariates that are outside the interquartile range.

{\scshape Selected Sites when $k=2$:} When $k=2$, the $k$-median solution is to again pick Dhaka-Pabna and additionally Dhaka-Pirojpur. 
Figure \ref{subfig:covariates_k_2} presents the covariates of Dhaka-Pabna (filled, blue circle) and Dhaka-Pirojpur (hollow, blue circle). While four covariate values of Dhaka-Pirojpur
are outside the interquartile range, they still appear more central than Dhaka-Gaibandha and Dhaka-Noakhali. However, relative to Dhaka-Pabna, the solution for $k=2$ adds a considerably less central site; Figure \ref{fig:hh_distance} illustrates that this site is a good nearest neighbor for some sites that would not otherwise be well matched.    

Figure \ref{subfig:covariates_sps_gechter_2} presents the corresponding solutions of \cite{egamidesigning} (squares) and \cite{gechter2023site} (triangles).\footnote{\cite{gechter2023site} impose an additional constraint on purposive sampling schemes: they require the two migration corridors selected for experimentation to have origins in different divisions. We note that both the solutions of \cite{egamidesigning} and the $k$-median problem satisfy this constraint as well.} The solution of the $k$-median problem and synthetic purposive sampling are no longer the same. We also note that the two sites selected by synthetic purposive sampling differ from the site selected by the same procedure when $k=1$. 

\begin{figure}[h!]
    \begin{center}
	\begin{subfigure}{.85\linewidth}
		\centering
		\includegraphics[width=\linewidth]{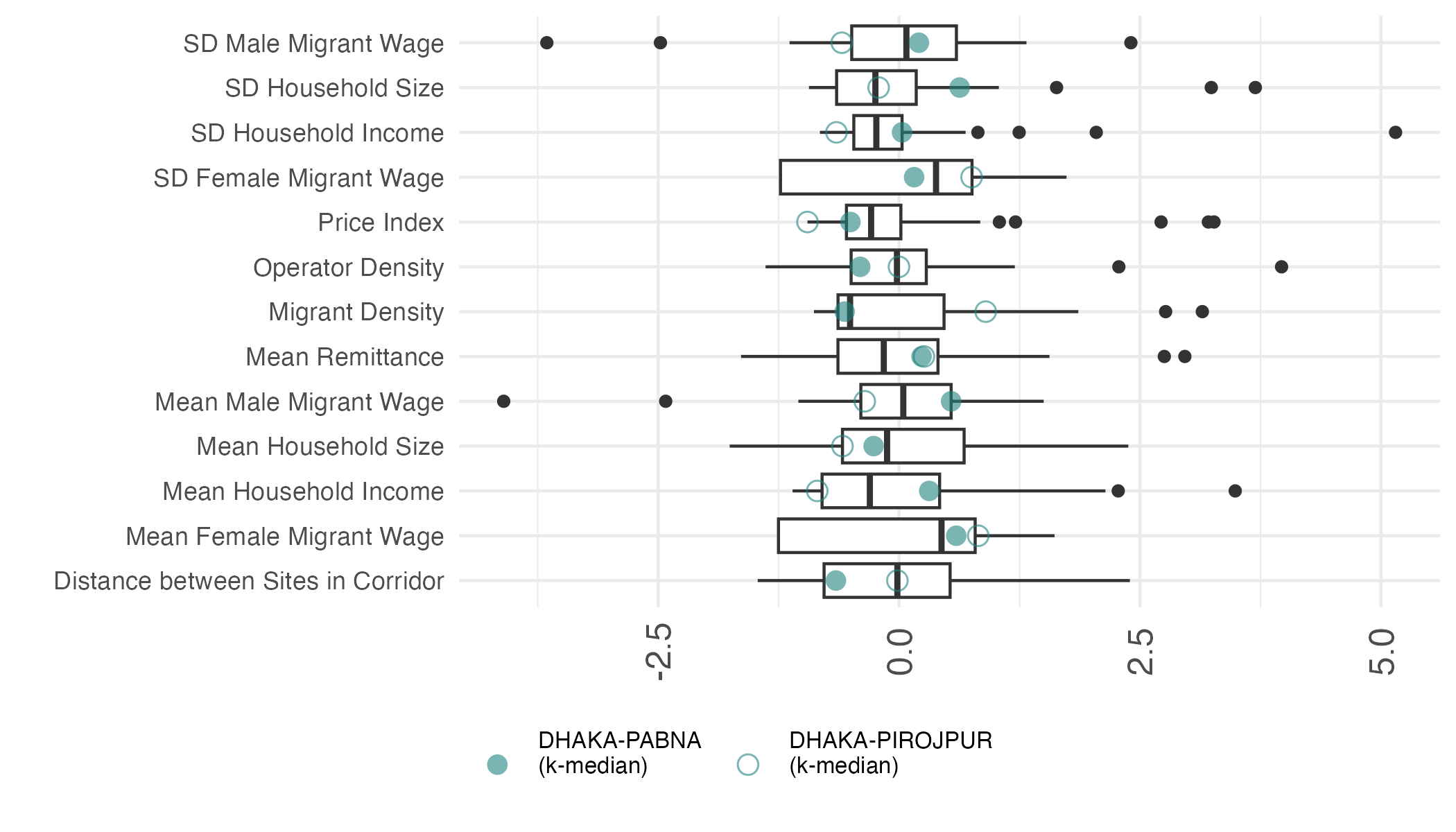}
		\caption{$k$-median}
		\label{subfig:covariates_k_2}
	\end{subfigure}
	\begin{subfigure}{.85\linewidth}
		\centering
		\includegraphics[width=\linewidth]{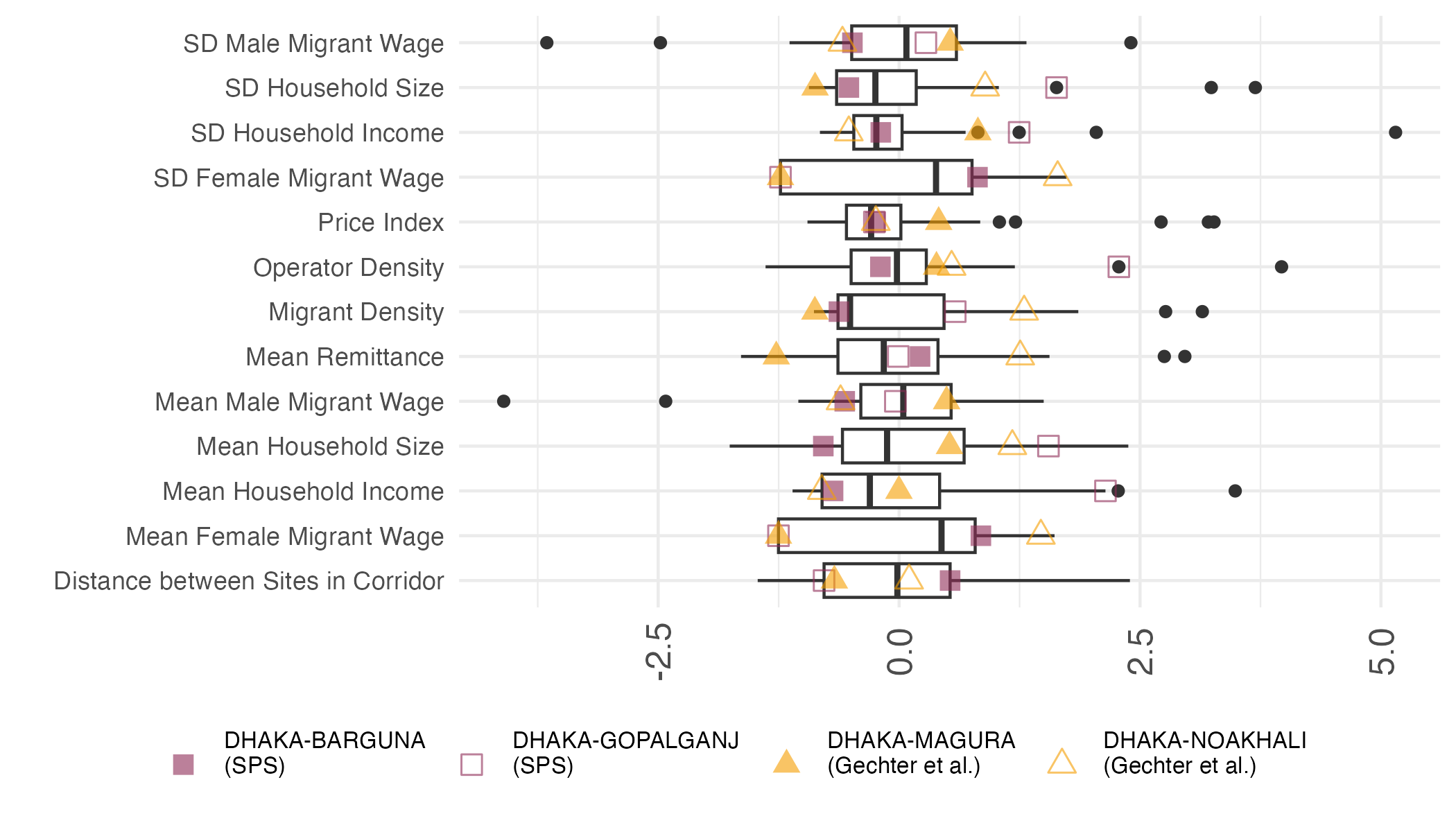}
		\caption{\cite{egamidesigning} and \cite{gechter2023site}}
		\label{subfig:covariates_sps_gechter_2}
	\end{subfigure}
	 \caption{Solution of the $k$-Median Problem ($k=2$) and Distribution of Site-level Covariates}
	 \label{fig:characteristics_gechter_2}
	\end{center}
 {\raggedright \footnotesize \textit{Notes}: Box plots for the distribution of covariates among migration corridors, constructed as explained in Figure \ref{fig:characteristics_gechter}. Each panel depicts the sites selected by the different approaches when $k=2$.}
\end{figure}

\begin{figure}[h!]
  \begin{center}
      \begin{subfigure}[b]{0.45\linewidth}
    \includegraphics[width=\linewidth]{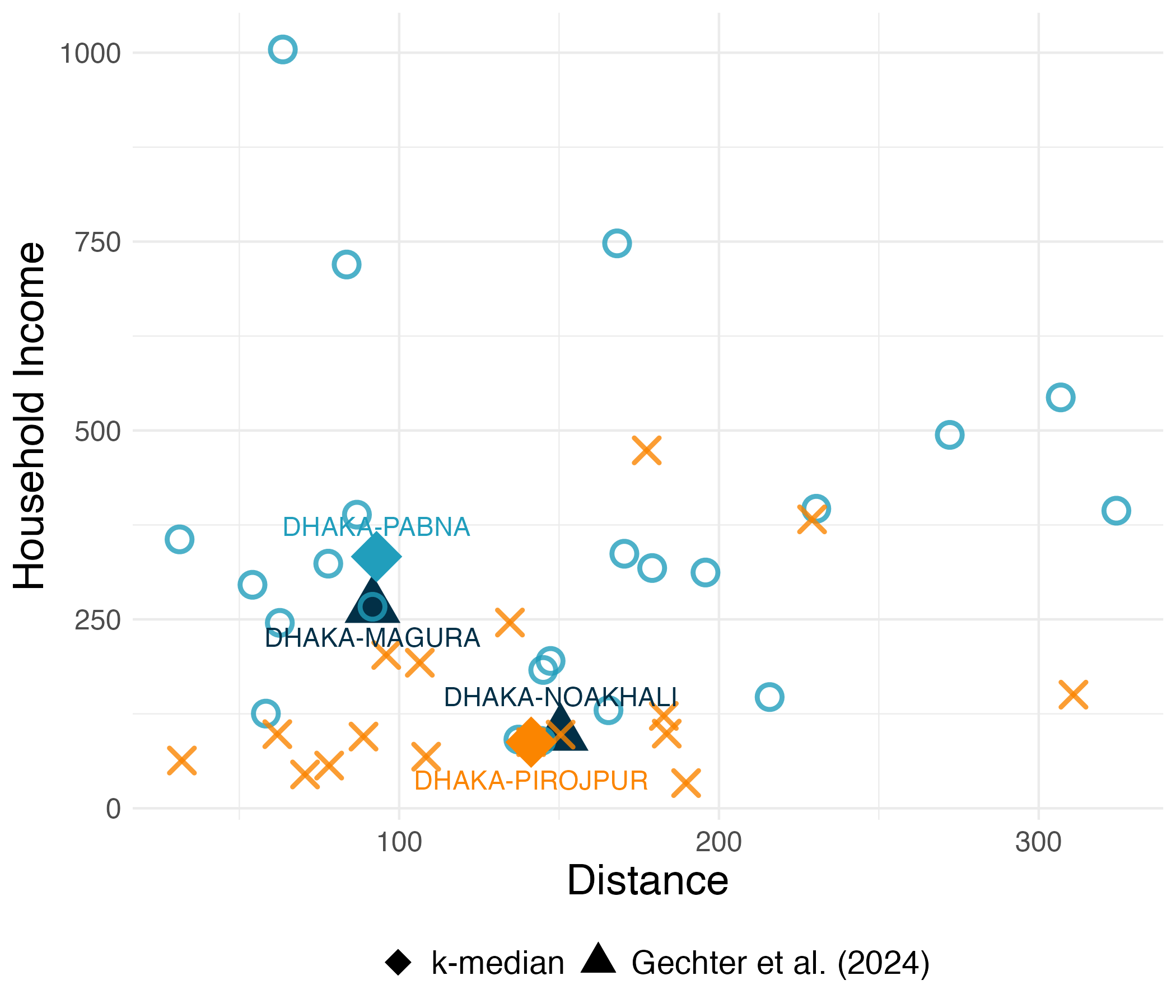}
    \caption{$k$-median} 
    \label{subfig:hh_distance_kmedian}
  \end{subfigure}
  \hfill \begin{subfigure}[b]{0.45\linewidth}
    \includegraphics[width=\linewidth]{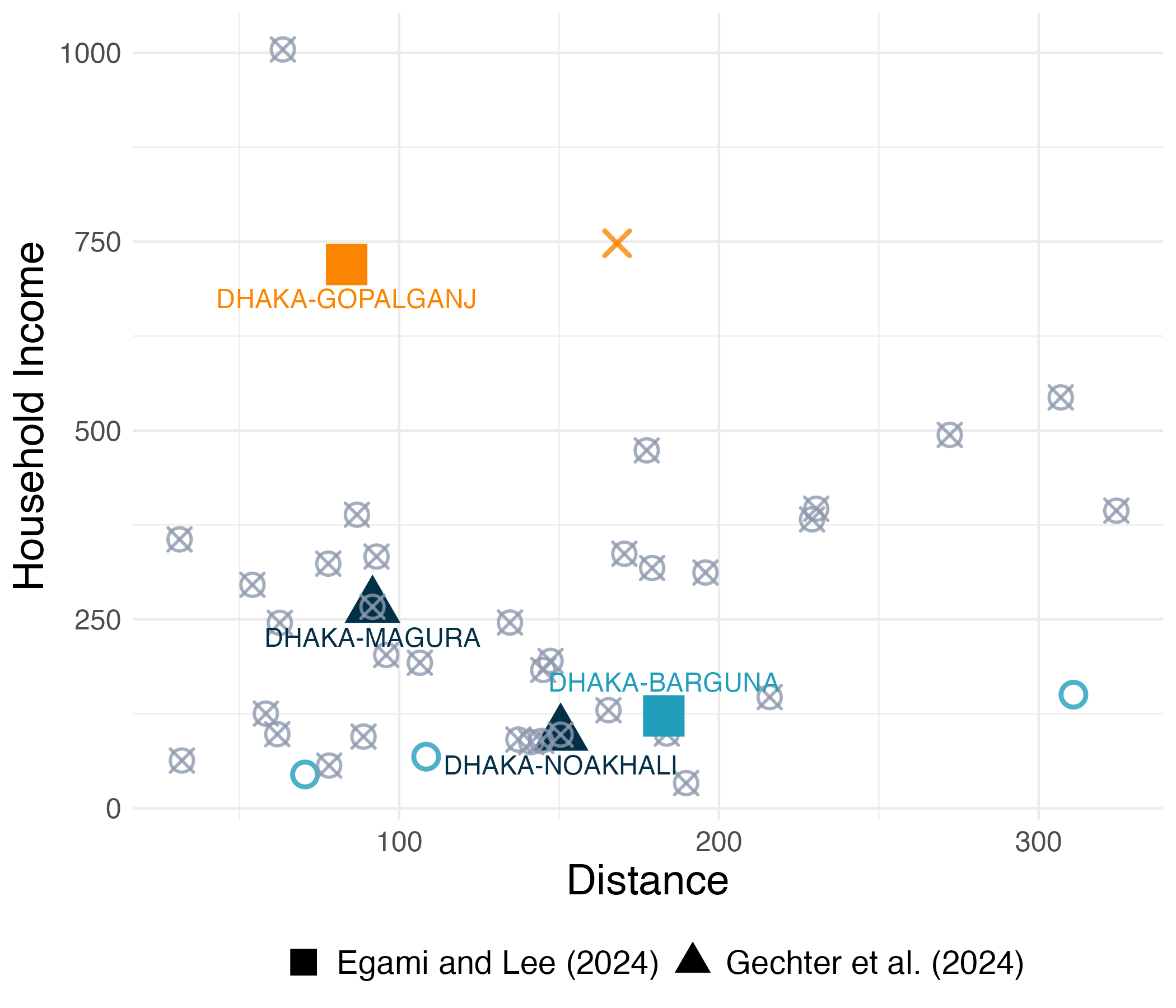}
    \caption{\cite{egamidesigning}} 
    \label{subfig:hh_distance_sps}
  \end{subfigure}
  \caption{Optimal Experimental Sites for $k = 2$ on a Two-dimensional Covariate Plane}
  \label{fig:hh_distance}
  \end{center}
{\raggedright \footnotesize \textit{Notes}: Each point represents a destination-origin migration corridor (site) on the two-dimensional covariate plane, using two covariates: distance between the two ends of each corridor and average household income in the home district (where the average is taken over households with a reported migrant). The solutions of each method are indicated by the shapes of the solid dots. In Panel \ref{subfig:hh_distance_kmedian}, blue circles denote policy sites with Dhaka-Pabna as the nearest neighbor, and orange crosses represent those with Dhaka-Pirojpur as the nearest neighbor. Similarly, in Panel \ref{subfig:hh_distance_sps}, blue circles and orange crosses follow the same coding, indicating that these policy sites rely on a single experimental site for constructing their synthetic control. Additionally, gray cross-circles represent policy sites that use both experimental sites for their synthetic control.} 
\end{figure}

Figure \ref{fig:hh_distance} presents a simple visualization of optimal connections between sites under different optimality criteria. The underlying scatter plots in both panels are the same; they visualize the location of corridors in (distance, household income)-space. Panel \ref{subfig:hh_distance_kmedian} indicates which sites were selected through solving the $k$-median problem and which of them any other site was matched with, where blue circles represent matches with Dhaka-Pabna and orange crosses represent matches with Dhaka-Pirojpur. Figure \ref{subfig:hh_distance_kmedian} also presents the sites selected by \cite{gechter2023site}, i.e. Dhaka-Magura and Dhaka-Noakhali, both marked with dark blue triangles. They appear close to the sites selected by the $k$-median problem, but in terms of that problem's criterion function, they only rank $455$ among $820$ candidate solutions and miss the problem's optimal value by $17\%$. Panel \ref{subfig:hh_distance_sps} similarly visualizes the sites selected by \citeauthor{egamidesigning}'s (\citeyear{egamidesigning}) synthetic purposive sampling approach, i.e. Dhaka-Gopalganj (orange square) and Dhaka-Barguna (blue square).\footnote{We generated this figure by using \citeauthor{egamidesigning}'s (\citeyear{egamidesigning}) code on \citeauthor{gechter2023site}'s (\citeyear{gechter2023site}) data.} We also visualize how policy sites are matched with experimental sites: orange crosses correspond to migration corridors that assign all of their weight to Dhaka-Gopalganj; blue circles correspond to migration corridors that assign all of their weight to Dhaka-Barguna; all other sites (gray crossed circles) assign strictly positive weights to both donor sites. This illustrates that these sampling schemes meaningfully differ. This would even be true if the synthetic purposive sampling approach were implemented but forcing degenerate (single donor) matches, because \citeauthor{egamidesigning}'s (\citeyear{egamidesigning}) approach would then reduce to a $k$-median problem but with \emph{squared} Euclidean distance as connection cost. Note also that we decided not to report how policy sites are matched with experimental sites in the Bayesian solution of \cite{gechter2023site}. This is simply because, under standard Gaussian process priors, the posterior mean for the treatment effect at each policy relevant site uses the information available for all experimental sites. It would be interesting to analyze whether treatment assignment rules that only use the information of some experimental sites could arise from sparsity-inducing priors such as those discussed in \cite{datta2016hierarchical}.

{\scshape Computational Cost:} The above examples are small enough so that the $k$-median problem could  easily be solved by brute-force enumeration of $41$ and $820$ candidate solutions, respectively. Needless to say, such an approach would not scale---for example, in this same application, $k=10$ induces $1,121,099,408$ candidate solutions.

Indeed, Figure \ref{subfig:Gurobi_vs_bruteforce} compares time to solve the $k$-median problem for $k \in \{1,\ldots, 10\}$ using a) the integer program formulation of the $k$-median problem in Section \ref{section:kmedian} and b) brute-force enumeration. The \texttt{MIP} solver in \texttt{Gurobi} solves all instances of the problem to provable optimality in less than one second each. In contrast, brute-force enumeration takes approximately $5$ hours for $k = 10$.\footnote{This was run in a Windows XPS with $10$ cores and $32$GB of RAM using R. The code is parallelized, and to avoid memory issues when $k > 8$, it runs a C++ function in the background to evaluate one combination at a time. This ensures that we are giving brute-force enumeration the best chance of success. } 

One potential benefit of brute-force enumeration is that one can check for multiple solutions, which actually occurred at $k=6$. In \texttt{Gurobi}, an ad hoc search could be conducted by modifying the \href{https://www.gurobi.com/documentation/current/refman/seed.html}{random number seed} or using the \href{https://www.gurobi.com/documentation/9.5/refman/concurrent_optimizer.html#sec:ConcurrentControl}{concurrent optimizer}, but discovery would not be guaranteed.

Figure \ref{subfig:Gurobi_vs_SPS} reports the time needed to implement the synthetic purposive sampling approach of \cite{egamidesigning}. This was done by using the \texttt{spsR} package with the option to use the \texttt{Gurobi} solver in the background. We consider it to be fast, taking less than $40$ seconds for $k = 10$. However, while synthetic purposive sampling can be formulated as a quadratic mixed integer program, $k$-median is linear and correspondingly faster to solve; indeed, solutions were almost instant for every $k$ up to $10$.

\begin{figure}[h!]
    \begin{center}
	\begin{subfigure}{.5\linewidth}
		\centering
		\includegraphics[width=\linewidth]{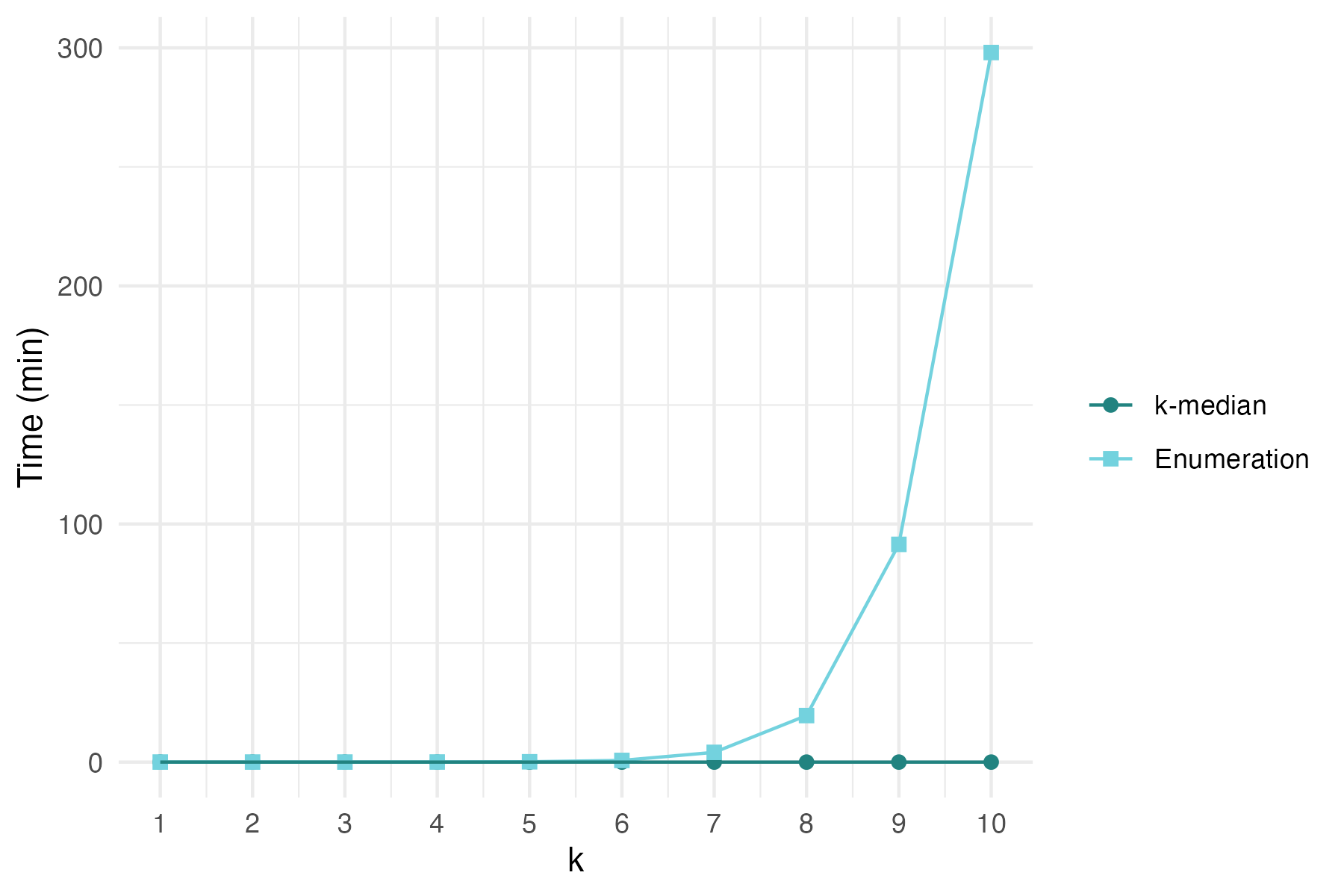}
		\caption{Integer Program vs. Enumeration}
		\label{subfig:Gurobi_vs_bruteforce} 
	\end{subfigure}\begin{subfigure}{.5\linewidth}
		\centering
		\includegraphics[width=\linewidth]{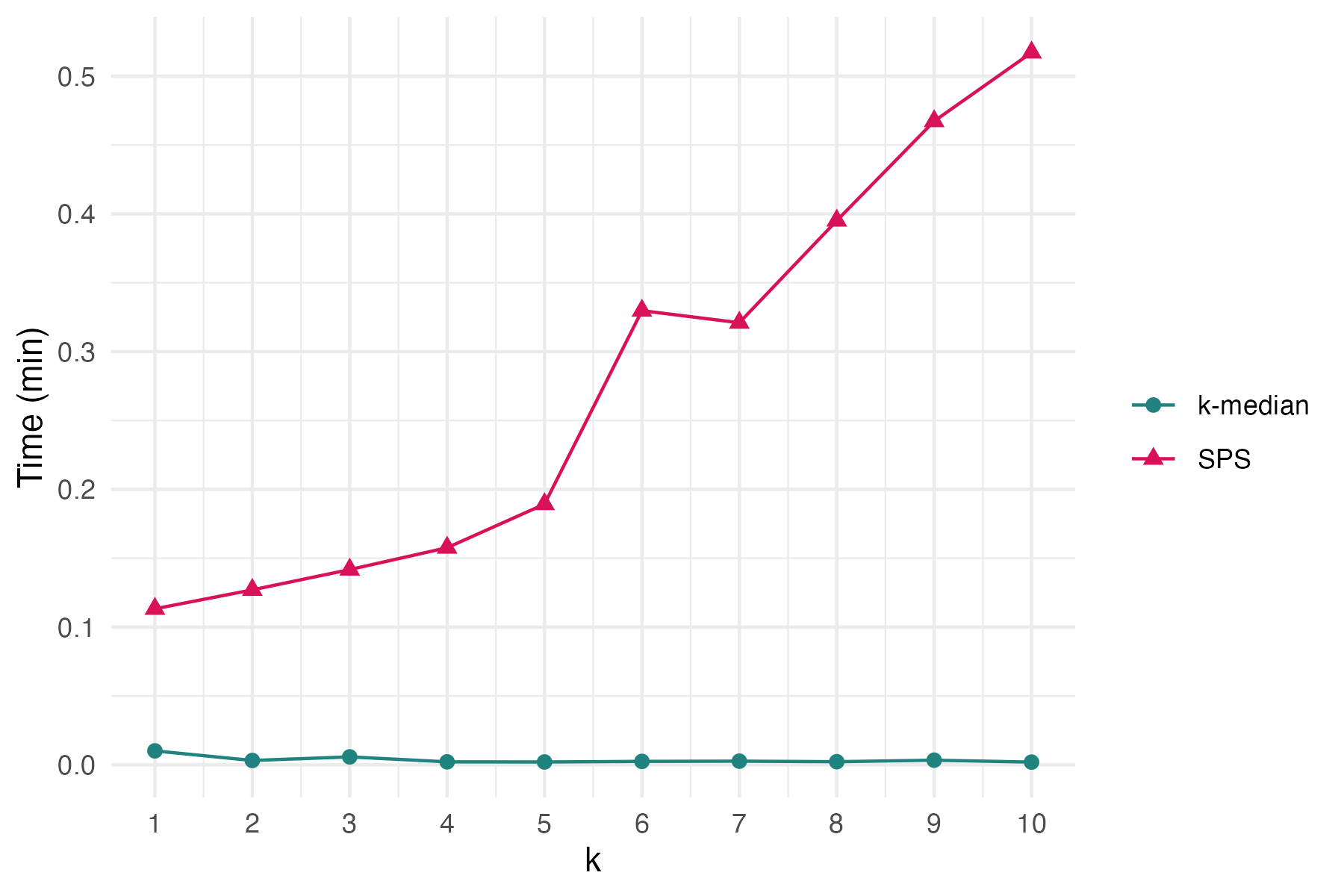}
		\caption{Integer Program vs. SPS}
		\label{subfig:Gurobi_vs_SPS}
	\end{subfigure}
	 \caption{Time Needed to Solve the $k$-Median Problem $k \in \{1,\ldots 10\}$}
	 \label{fig:runtime}
	\end{center}
 {\raggedright \footnotesize \textit{Notes}: Time comparison of different purposive sampling approaches. Time (vertical axis) is in minutes. The dark, blue line with circles represents the time it takes to solve the integer program in Section \ref{section:kmedian} using the \texttt{MIP} solver in \texttt{Gurobi} to provable optimality (Gurobi gives the solution in less than a second). The light, blue line with squares in Panel a) represents the time it takes to solve the $k$-median problem using brute-force enumeration. The red line with triangles in Panel b) represents the time needed to implement the synthetic purposive sampling approach using the \texttt{spsR} package.}
\end{figure}

In this application we model treatment effect heterogeneity as a function of site-level covariates. This is partly motivated by the fact that detailed individual-level data may not be available for all sites in empirical applications. If individual-level covariates were available, and if there were no heterogeneity of CATEs (defined via individual-level characteristics)  across sites, one could estimate CATEs for experimental sites and then reweight them to derive average treatment effects for policy-relevant sites; see, for example, the discussion in \citet[][p. 493]{list2024optimally}. In this case, many state-of-the-art methods for estimating ATE with experimental data could be employed. However, one might still be concerned about external validity, namely that those CATEs vary at the site level. A pure  reweighting method would then not work, and some extrapolation would still be required. Our approach already offers a practical solution to deal with the individual-level covariates, at least in the case in which there is a single policy-relevant site: One can estimate site-specific CATEs and take $\widehat{\tau}_s$ to be the average of these CATEs over the distribution of covariates in each of the experimental sites. Then we are back to the original set-up of our problem, in which we need to decide how to transport these estimators to make policy choices in the site of interest. Since the payoff relevant parameter continues to be the site-level treatment effect, we think it is reasonable to model treatment effect heterogeneity as a function of the site-level characteristics.

\subsection{Multi-Country Survey Experiments in Europe} 
\label{sec:multicountry}

Our second application revisits a multi-country survey experiment originally conducted and analyzed in \cite{naumann2018attitudes} and discussed in \cite{egamidesigning}. The question of interest is whether native-born inhabitants of a particular country are more supportive of immigration depending on whether the potential migrants are high-skilled or low-skilled. \cite{naumann2018attitudes} carried out a survey experiment in $15$ European countries listed in Figure \ref{fig:heatmap}. Respondents were randomly assigned to report their attitudes towards either high-skilled (``treatment'') or low-skilled (``control'') immigrants.

While the experiments have already been conducted and outcomes of each experiment are available for all countries, we follow \cite{egamidesigning} and consider the situation of a researcher who can only conduct $k=6$ experiments. We let all $15$ countries be both potential experimental and policy-relevant sites. That is, $\#\mathcal{S}_E= \#\mathcal{S}_P  = 15$.\footnote{A potential situation we have in mind is one of a researcher who would like to give policy advice to the policymakers in these countries on whether to initiate an immigration reform that could favor either high-skilled or low-skilled immigrants. The researcher knows that policymakers are interested in voters' attitudes towards these types of reforms. We assume that the researcher is only able to experiment in a subset of countries due to administrative or budget constraints, and that he/she needs to extrapolate voters' attitudes of the other policy-relevant sites based on the experimental estimates. The researcher needs to decide whether to recommend the implementation of an immigration reform directed to either high-skilled or low-skilled immigrants.}  

The only data needed to solve the $k$-median problem are site-level covariates of both experimental and policy-relevant sites. We use the same covariates as \cite{egamidesigning}.\footnote{These are: gross domestic product (GDP), size of migrant population, unemployment rate, proportion of females, mean age, mean education, baseline level of support for immigration by the general public, and a categorical variable that indicates the subregions in Europe (i.e., South, North, East, and West). The covariate data can be accessed in the open-source software R package, \href{https://naokiegami.com/spsR}{\texttt{spsR}}.}
They are on different scales: for example, GDP is measured in 2015 U.S. dollars, while the unemployment rate is reported in percentage points. Following standard practice for clustering problems and also \citet{egamidesigning}, we standardize all of them. We construct the matrix $C$ the same way as in the previous application. Thus, $\mathcal{S}$ contains two copies of each site, with the first $15$ sites comprising $\mathcal{S}_E$ and the next $15$ sites (in the same order) comprising $\mathcal{S}_P$. For any $i \in \mathcal{S}_{E}$ and $j \in \mathcal{S}_{P}$ with $j \neq (15+i) $, we set $C_{i,j}\propto\| X_i - X_{j-15} \|$. If $j=15+i$, we set $C_{i,15+i}\propto\min_{i' \in \{1,\ldots, 15\} \backslash i }\| X_{i'} - X_i \|(1-\epsilon)$ with $0\leq\epsilon<1$. Finally, we set $C_{i,i}=0$ for all $i\in\{1,...,30\}$.

For $k=6$ and $\epsilon$ ranging from $0.01$ to $0.10$ in increments of $0.01$, the $k$-median problem is consistently solved by \{Denmark, France, Ireland, Slovenia, Sweden, Switzerland\}. Figure \ref{fig:covariates_multicntry} visualizes the distribution of standardized covariates, along with the sites selected by both the $k$-median and the synthetic purposive sampling approach. Two of the six selected sites (Switzerland and Denmark) are common to this particular $k$-median solution and the synthetic purposive sampling approach.

\begin{figure}[h!]
	\begin{center}
	    \begin{subfigure}{.8\linewidth}
		\centering
		\includegraphics[width=\linewidth]{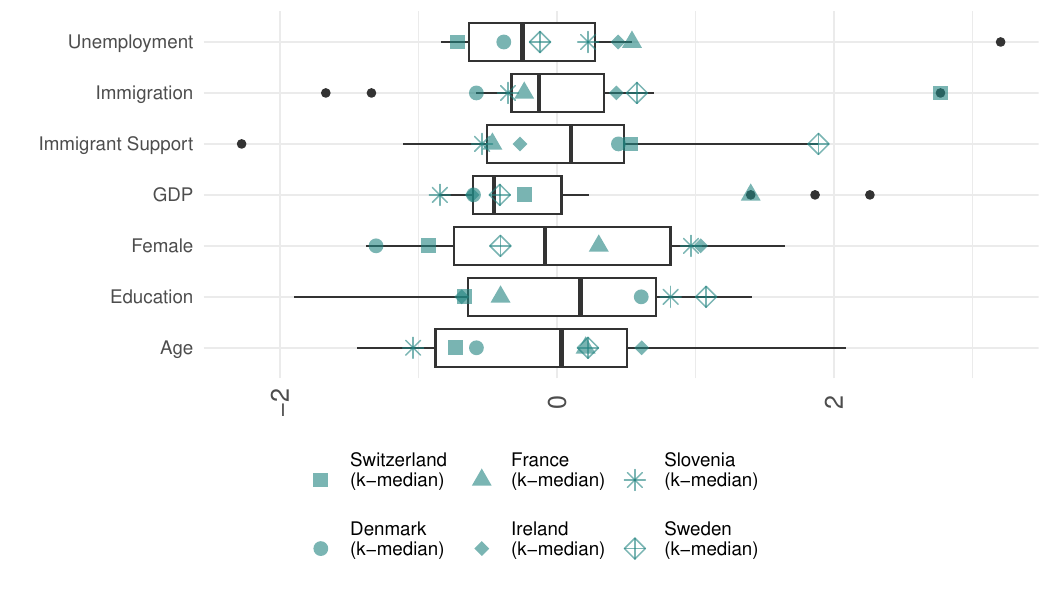}
		\caption{$k$-median}
		\label{subfig:boxplot_k6_kmedian_multicntry}
	\end{subfigure}
	\begin{subfigure}{.8\linewidth}
		\centering
        \includegraphics[width=\linewidth]{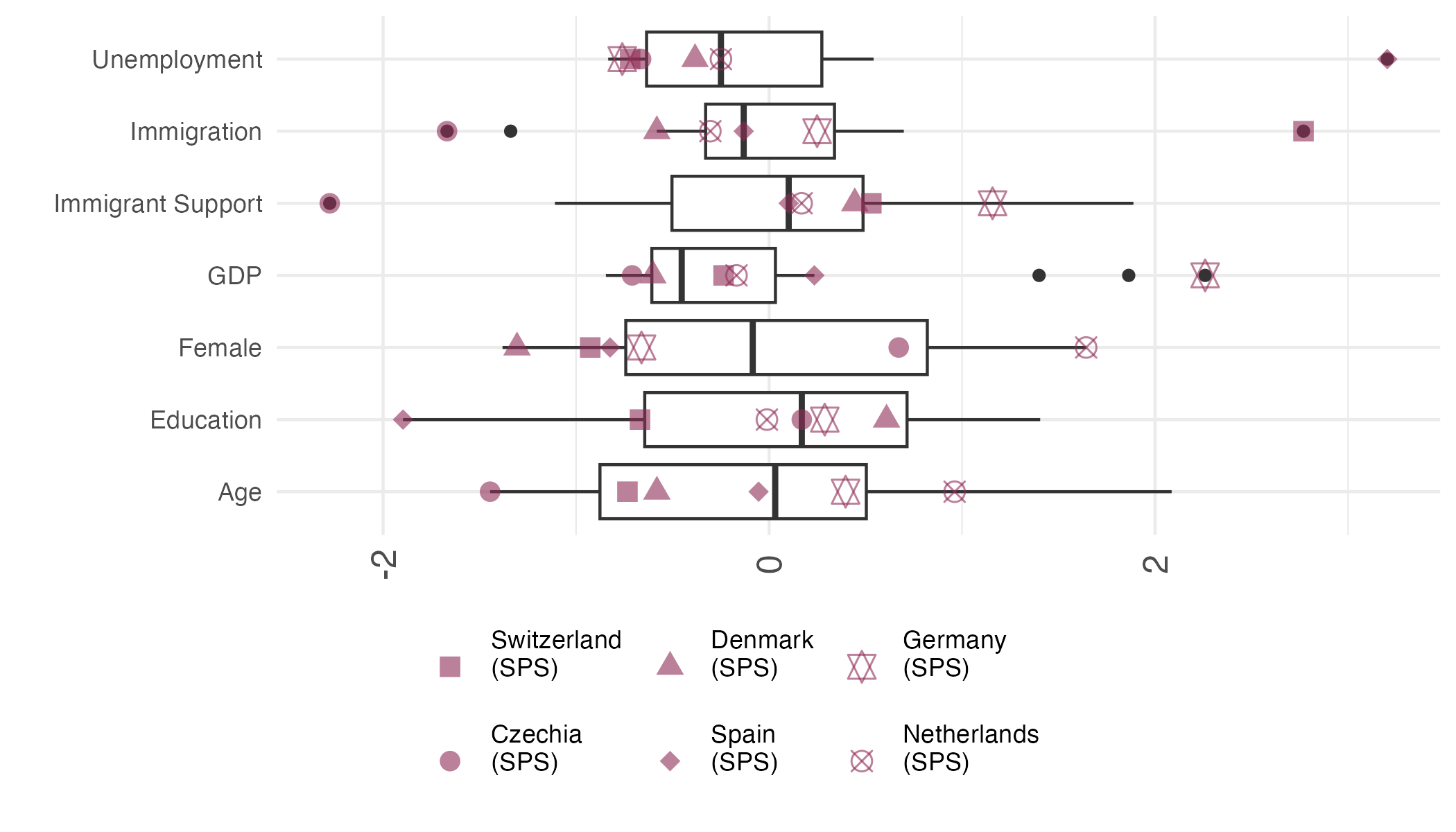}
		\caption{\cite{egamidesigning}}
		\label{subfig:boxplot_k6_sps_multicntry}
	\end{subfigure}
	\caption{Distribution of Site-level Covariates in The Multi-Country Survey Experiment $(k = 6)$}
	\label{fig:covariates_multicntry}
	\end{center}
{\raggedright \footnotesize \textit{Notes}: Box plots showing the distribution of covariates among the fifteen countries. The box plots are constructed as explained in Figure \ref{fig:characteristics_gechter}. Each panel depicts the sites selected by the different approaches when $k=6$. 
}   
\end{figure}

Figure \ref{fig:heatmap} visualizes the weights that each policy-relevant site places on the selected experimental sites. In both panels, the $x$-axis displays the union of experimental sites selected under the two solutions, with grey-shaded columns indicating sites that are not selected under the solution shown in the corresponding panel. The color intensity of each cell, ranging from white (weight $0$) to dark red (weight $1$), represents the weight assigned by a policy-relevant site to the corresponding experimental site. In panel \ref{subfig:heatmap-kmedian}, each row contains exactly one dark red cell, which means each policy-relevant site assigns all the weight to exactly one experimental site, i.e. its nearest neighbor; for example, the United Kingdom uses only the information of France. The heatmap in Figure \ref{subfig:heatmap-sps} is considerably more dense, with each policy-relevant site that is not selected for experimentation assigning positive weights to multiple experimental sites. For instance, the synthetic experiment for the United Kingdom assigns positive weights to the Czech Republic, Netherlands, and Germany. The relative color intensities further illustrate the magnitudes of these weights: for the United Kingdom, the Netherlands receives the largest weight ($0.74$), while the Czech Republic receives the smallest ($0.01$).

Figure \ref{subfig:heatmap-kmedian} also shows that Switzerland and Sweden, as two selected countries in the $k$-median problem, have single dark red cells only in the rows corresponding to themselves, suggesting that they were selected because no other country provides a close enough match for them. In contrast, in Figure \ref{subfig:heatmap-sps}, Switzerland receives positive weights from six other countries.
\begin{figure}[!ht]
	\begin{center}
	    \begin{subfigure}{.5\linewidth}
		\centering
		\includegraphics[width=\linewidth]{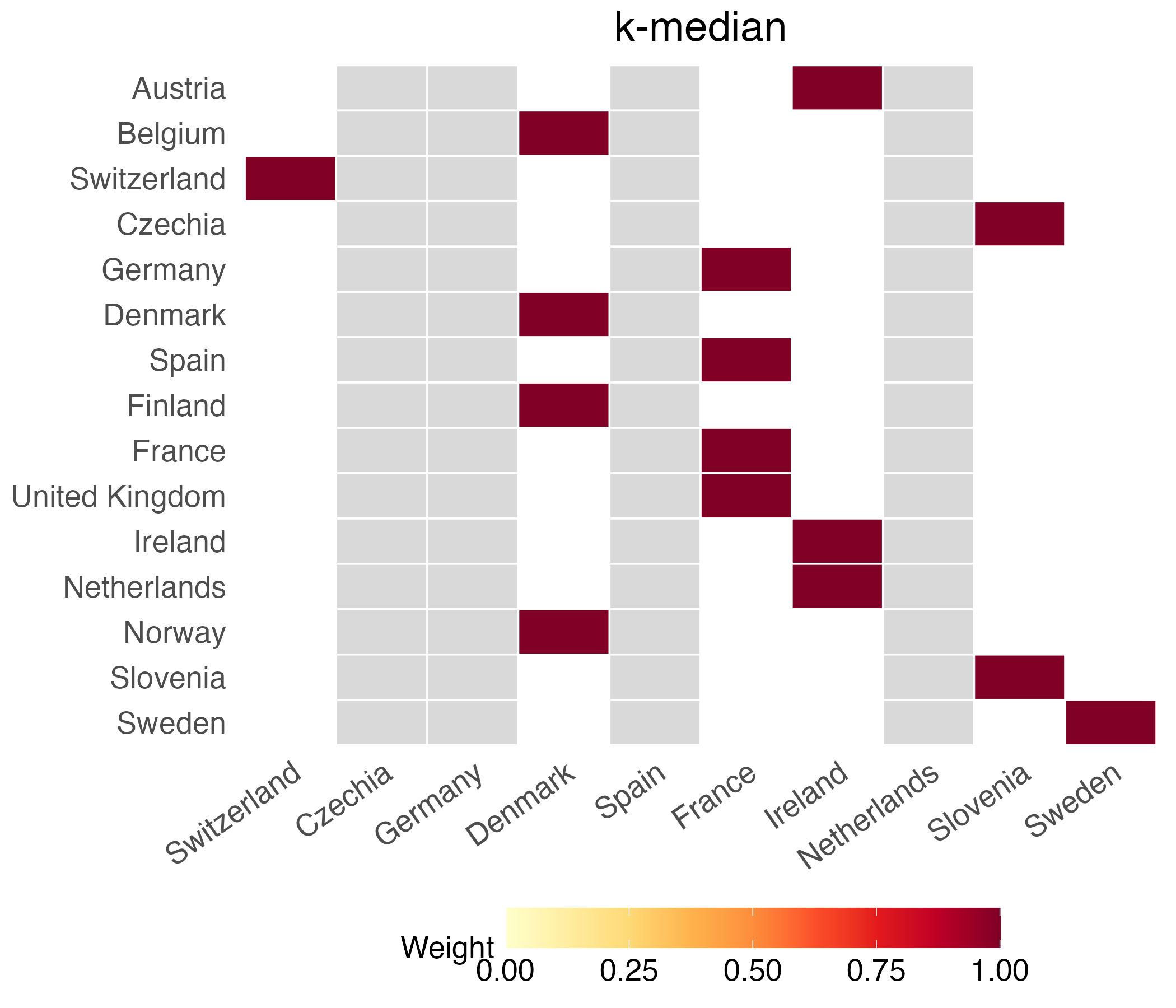}
		\caption{$k$-median}
		\label{subfig:heatmap-kmedian}
	\end{subfigure}\begin{subfigure}{.5\linewidth}
		\centering
        \raisebox{-2in}{\includegraphics[width=\linewidth]{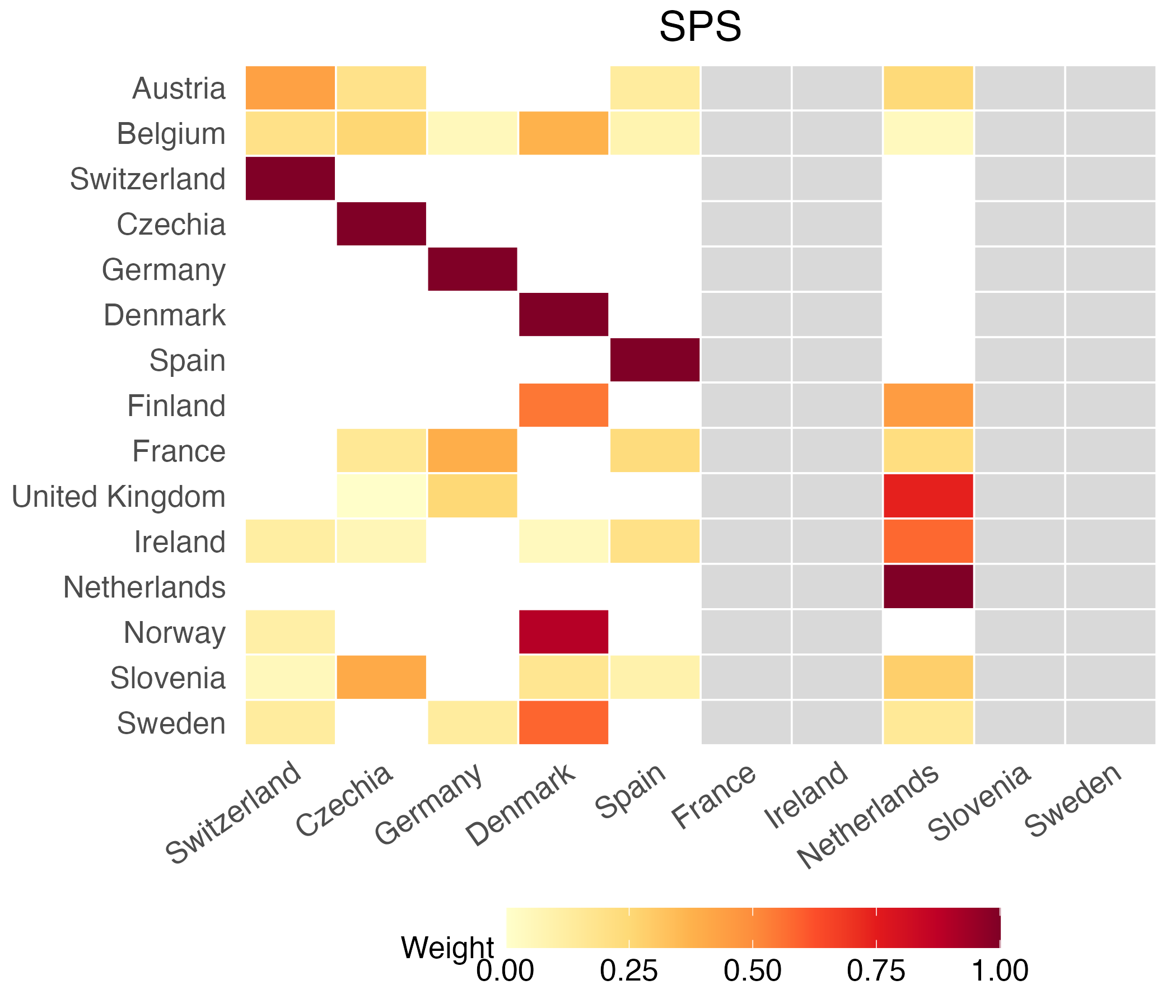}}
		\caption{\cite{egamidesigning}}
		\label{subfig:heatmap-sps}
	\end{subfigure}
	\caption{Heatmap Plot for $k=6$}
	\label{fig:heatmap}
	\end{center}
{\raggedright \footnotesize \textit{Notes}: Each panel is a heatmap showing, for every policy site (row), the weight it places on each selected experimental site (column). The $x$-axis in both panels lists the union of the selected experimental sites under the two solutions. Columns shaded grey are experimental sites that were not selected under the reported solution. In Panel \ref{subfig:heatmap-kmedian}, each row has weight one on a single non-grey column, which corresponds to the experimental site that is the nearest neighbor of the policy site. In Panel \ref{subfig:heatmap-sps}, each row may place positive weight on one or more non-grey columns, indicating that the corresponding policy site uses a weighted average of experimental sites to construct its synthetic control. The color intensity of each cell, from white (weight $0$) to dark red (weight $1$), indicates the weight that the policy site puts on the corresponding experimental site.}   
\end{figure}

\section{Extensions} \label{section:extensions}

\subsection{Fixed Costs of Experimentation} \label{subsection:fixed_costs}

We next allow for the possibility that running an experiment in a site $s \in \mathcal{S}_{E}$ has a fixed cost $c_s$. This means that the welfare of a decision rule $T$, given that sites $\mathscr{S}$ are selected for experimentation, corresponds to
\begin{equation*} \label{eqn:welfare_cost}
\mathcal{W}_c(T,\mathscr{S},\tau) := \frac{1}{\# \mathcal{S}_{P}} \left( \sum_{s \in \mathcal{S}_{P}} \tau_s \mathbb{E}_{\tau_{\mathscr{S}}}[T_{s}(\widehat{\tau}_{\mathscr{S}})] -  \sum_{s \in \mathscr{S}} c_s \right).
\end{equation*}
Based on this welfare function, the oracle action for the policymaker is to implement the policy in any policy-relevant site $s$ for which $\tau_s \geq 0$. Expected regret of $(T,\mathscr{S})$ becomes  
\begin{equation*} \label{eqn:regret_cost}
\mathcal{R}_c(T,\mathscr{S},\tau) := \frac{1}{\# \mathcal{S}_{P}} \sum_{s \in \mathscr{S}} c_s + \frac{1}{\# \mathcal{S}_{P}} \sum_{s \in \mathcal{S}_{P}} \tau_s \left( \mathbf{1}\{ \tau_s \geq 0\} - \mathbb{E}_{\tau_{\mathscr{S}}}[T_{s}(\widehat{\tau}_{\mathscr{S}})] \right).
\end{equation*}
Under the assumptions of Theorem \ref{thm:main}, one can show that solving  
\[ \inf_{\mathscr{S} \in \mathcal{A}(k)} \left( \inf_{T \in \mathcal{T}_{\mathscr{S}}} \sup_{\tau \in \tau(C)}  \mathcal{R}_c(T,\mathscr{S},\tau) \right) \]  
is equivalent to solving
\begin{equation} \label{eqn:k-facility-location} 
\inf_{\mathscr{S} \in \mathcal{A}(k)}\left(2\sum_{s \in \mathscr{S}}  c_s +   \sum_{s \in \mathcal{S}_{P} }   C_{N_{\mathscr{S}}(s),s}  \right).
\end{equation}
The problem in \eqref{eqn:k-facility-location} is the \emph{metric uncapacitated $k$-facility location problem}.\footnote{Uncapacitated means that there is no capacity constraint on the number of clients that each facility can accommodate. See \cite[Chapter 4.5]{shmoys2011design} and \cite{zhang2007new}.} Just as in the $k$-median problem, there is a set of facilities $\mathcal{F}$ and a set of clients (or cities) $\mathcal{C}$. Now we assume that there is an opening cost $c_i \in \mathbb{R}_{+}$ associated with each facility $i \in \mathcal{F}$. The connection cost between facilities and clients is as before. The goal is to open at most $k$ facilities and connect each client to one facility so that total cost is minimized. This is a common extension of the $k$-median problem where there is a fixed cost of opening each facility. It can also be formulated as a linear integer program, namely 
\begin{align*}
    \min_{\{y_i,x_{i,j}\}_{i \in \mathcal{S}_{E}, j \in \mathcal{S}_{P}}} & \left( \sum_{i \in \mathcal{S}_E} y_i c^*_i  +  \sum_{i \in \mathcal{S}_E, j \in \mathcal{S}_P } x_{i,j}\cdot c(i,j) \right) \\
\text{such that } &\sum_{i \in \mathcal{S}_E} x_{i,j} = 1, \;\;\;\;\; \forall j \in \mathcal{S}_P,  \\
    &\sum_{i \in \mathcal{S}_E} y_i \le k,\\
    & 0 \le x_{i,j} \le y_i, \;\;\;\;\; i \in \mathcal{S}_E, j \in \mathcal{S}_P,\\
    &  y_i \in \{0,1\}, x_{i,j} \in \{0,1\}, \;\;\;\;\; i \in \mathcal{S}_E, j \in \mathcal{S}_P. \end{align*}
Just as before, the choice variable $y_i$ indicates whether facility $i$ is open, and $x_{i,j}$ indicates whether client $j$ is assigned to facility $i$. Constraints are imposed to ensure that each client $j$ is assigned to at least one facility, that there are no more than $k$ facilities in total, and that a client can only be assigned to an open facility. The connection cost between a facility $i$ and a client $j$ is still given by $C(i,j)$, but now there is a fixed cost $c^*_i \equiv 2 c_i $ of opening a facility $i$. Consequently, every time a new site is considered for experimentation, there is a trade-off between its contribution to reduce regret and the fixed cost of experimentation. When all of the costs $c_s$ are small, the fixed costs in the objective function become negligible, and the solution of the uncapacitated facility location problem can be approximated by the solution of the $k$-median problem.  

So far, we treated the number of sites selected for experimentation as exogenous, but the Option C framework of \cite{list2024optimally,list2026economics} suggests endogenizing it. We argue that solving the uncapacitated $k$-facility location problem in \eqref{eqn:k-facility-location} for different values of $k$ provides a simple strategy to endogenize the choice of $k$. Note that the objective function in the uncapacitated $k$-facility location problem has two components that will behave differently as a function of $k$. One of the components in the objective function captures the connection cost; i.e., the worst-case distance from each policy-relevant site to its nearest neighbor. This component decreases as a function of $k$ (the more sites we add, the closer the corresponding nearest neighbors will be). The other component captures the fixed costs of experimentation (which we have allowed to be heterogeneous across sites). If we assume that exactly $k$ sites are selected for experimentation, this component will typically be increasing as a function of $k$. The easiest way to see this is to consider the case in which the costs are homogeneous and equal to $c$, in which case the total cost will be $k \cdot c$. Thus, one can solve the uncapacitated $k$-facility location problem (enforcing the selection of exactly $k$ sites) from $k \in \{1, \ldots, \#\mathcal{S}_{E}\}$ and choose the value of $k$ that gives the smallest objective function.

Note that the calculations suggested above are designed to minimize worst-case regret (as a function of $k$). Regret depends both on the expected benefit and cost of experimentation. But expected benefits depend on the true (and unknown) treatment effects. As we have shown, taking a worst-case perspective allows us to get a sharp conclusion about the selection of experimental sites (for each fixed $k$). Moreover, the value of this worst-case regret (which now includes the costs of experimentation) then allows us to optimally choose $k$. An alternative calculation could be based on Bayesian calculations that place priors on the treatment effects as in \cite{gechter2023site}.

\subsection{Random Selection of Experimental Sites}\label{subsec:randomize}
Our analysis so far focused on minimizing the worst-case (welfare-based) regret among \emph{purposive} sampling schemes, i.e. excluding randomized (including nonuniformly randomized) such schemes. This is an important limitation---a widespread view among experimenters is that \emph{``the external validity of randomized evaluations for a given population (say, the population of a country) would be maximized by randomly selecting sites and, within these sites, by randomly selecting treatment and comparison groups''} \citep[][p. 3953]{duflo2007using}. A similar point is made in \cite[p. 493]{list2024optimally}: \emph{``if the researcher chooses locations at random in an initial stage of the experimental design, this will lead to generalizable results across all potential locations.''}  

\subsubsection{Change of Timing for the Statistical Game}\label{subsec:timing}
We now extend our baseline framework (which excludes fixed costs of experimentation) to allow for randomized site selection. First, we note that whether randomization is potentially desired delicately depends on the decision-theoretic setup. For example, the optimal sampling scheme in the Bayesian setting of \citet{gechter2023site} will typically be purposive. Whether randomization improves minimax regret depends on how exactly the decision problem is formulated, which can be related to the timing assumptions in a game between the policy maker and an adversarial Nature. We will first clarify this observation and then provide a brief illustration of randomized solutions.

To formalize the discussion, let $M$ denote the cardinality of $\mathcal{A}(k)$. Let $\Delta \left( \mathcal{A}(k) \right)$ denote the set of all probability distributions over the $M$ 
elements of $\mathcal{A}(k)$. We define a randomized site selection as a probability distribution $p := (p_1, \ldots, p_{M}) \in \Delta \left( \mathcal{A}(k) \right)$. The econometrician will pick a subset from the experimental sites by drawing one realization of a distribution $p^* \in \Delta \left( \mathcal{A}(k) \right)$ that she specified. Denote the randomly selected sites by $\mathscr{S}^*$.

Whether the decision maker will want to randomize crucially depends on timing assumptions, namely on when Nature can move. Consider first the case in which Nature may choose $\tau$ after seeing the realization of $\mathscr{S^*}$ (and knowing $\mathcal{T}_{\mathscr{S}^*}$). Then, the risk of using treatment rule $T \in \mathcal{T}_{\mathscr{S}^*}$ is
\[ \mathcal{R}(T,\mathscr{S}^*,\tau)  \]
and the worst-case payoff becomes 
\[  \sup_{\tau \in \tau(C)}  \mathcal{R}(T,\mathscr{S}^*,\tau).  \]

The minimax problem faced by the econometrician after $\mathscr{S}^*$ has been realized is:
\[  V( \mathscr{S}^* ) : = \inf_{T \in \mathcal{T}_{\mathscr{S}^*}} \sup_{\tau \in \tau(C)}  \mathcal{R}(T,\mathscr{S}^*,\tau).  \] 

With slight abuse of notation, let $p(\mathscr{S}^*)$ denote the probability of choosing $\mathscr{S}^*$ at random under $p \in \Delta \left( \mathcal{A}(k) \right)$. The (ex ante) expected payoff of any randomized site selection is
\[  \sum_{ \mathscr{S}^* \in \mathcal{A}(k) } p(\mathscr{S}^*) \cdot V(\mathscr{S}^*),   \]
and the optimal randomized site selection solves 
\[  \inf_{p \in \Delta(\mathcal{A}(k))} \sum_{ \mathscr{S}^* \in \mathcal{A}(k) } p(\mathscr{S}^*) \cdot V(\mathscr{S}^*) .  \]
But this problem is solved by any $p$ supported on $\arg \min_{\mathscr{S}} V(\mathscr{S})$, the set of purposive sampling schemes that solve the site selection problem. If that problem's solution is unique, the policymaker will \emph{never} randomize under this timing of the game. Furthermore, this ``sequential'' timing may feel natural in applications that we have in mind.

That said, the timing suggested by \citeauthor{Wald1950}'s (\citeyear{Wald1950}) original application of the minimax principle is one in which the policymaker commits to both a randomized sampling scheme and a set of contingent (on sampling scheme) decision rules and Nature picks $\tau$ before any randomization was realized.
In this case, randomization might be optimal.
Indeed, Appendix \ref{subsection:proof_randomize} presents an extremely stylized example (there are only two sites; their treatment effects may differ by up to a known number; the policymaker can choose exactly one site to sample; there are no voltage effects and there is no sampling uncertainty, so that one site's treatment effect will be learned perfectly) in which exact worst-case regret from purposive sampling exceeds exact MMR by $60\%$. However, the example is carefully rigged to work against purposive sampling (the signal is very informative, half of all policy sites get sampled, and there is no interesting spatial pattern of heterogeneity) and is one in which the first-best solution is obvious from symmetry.

In related work, \cite{fernandez2024} explore the use of the \emph{Hedge Algorithm} for numerically finding the MMR optimal randomized selection scheme in Section \ref{subsection:bangla}'s application. Their results suggest that, even if randomization is allowed, choosing the most representative site may be approximately optimal. Moreover, in the application they consider, the approximately optimal randomizations are far from uniform. 

\subsubsection{Special Set of Nearest Neighbors}\label{sec:special.set}
We next present a sufficient condition under which the solution of our $k$-median problem will be optimal even when randomized site selection is allowed and under both timing assumptions just discussed.
\begin{theorem}\label{thm:no.random}
Suppose the conditions of Theorem \ref{thm:main} hold. In addition, suppose there exists
a set $\mathcal{N}\in\mathcal{A}(k)$ such that $\min_{s' \in \mathcal{N}}C_{s',s} =\min_{s' \in \mathcal{S}_E}C_{s',s}$,
$\forall s\in\mathcal{S}_{P}$. Then the solution of the $k$-median
problem \eqref{eqn:k-median} is MMR optimal even when randomized site selection is allowed. 
\end{theorem}
 \begin{proof}
See Appendix \ref{section:proof_no_random}
\end{proof}
Theorem \ref{thm:no.random} says that,  if there exists a special set $\mathcal{N}\in \mathcal{A}(k)$ that contains a nearest neighbor in $\mathcal{S}_E$ of each policy-relevant site, then focusing on purposive sampling is not going to bind and  does not leave welfare on the table.  We think this result is particularly useful for practitioners as it offers a diagnostic tool to assess whether randomized site selection is likely to improve upon the $k$-median solution, namely when a policy-relevant site has a nearest  neighbor in $\mathcal{S}_E$ that is  not among the selected sites (although note that Theorem \ref{thm:no.random} only provides a sufficient, not necessary condition).

However, even if a policy-relevant site's nearest  neighbor in $\mathcal{S}_E$ is not among the selected sites, the improvement from randomization need not be large. Consider our first application: Clearly, there is no site that can be a $1$-universal nearest neighbor. However, there is a sense in which the site selected by the $k$-median approach is closer to most policy-relevant sites than i) the site selected in \cite{gechter2023site} and ii) the site in which the \cite{lee2021poverty} conducted the original experiment. A simple visualization of this when $k=1$ is given in Figure 3 of \cite{fernandez2024}, where the distances from the 38 sites (excluding the 3 sites we just discussed) are reported. Moreover, \cite{fernandez2024} find that the improvement from randomly selecting any of the 3 candidate solutions we just discussed is limited. We think that a figure analogous to Figure 3 of \cite{fernandez2024} (for the $k$-median solution and a few other candidate sites) could be useful to get a sense of the potential improvements of randomly selecting the candidate sites under consideration.

\subsection{Other Forms  of Treatment Heterogeneity}\label{sec:other.resitrction}

\subsubsection{Voltage Effects Based on Observed Covariates}\label{subsec:other.distance}

To see how our results can allow for voltage effects based on observed covariates beyond Example 1, let $m(x,x^{\prime}):\mathbb{R}^{d}\times\mathbb{R}^{d}\rightarrow\mathbb{R}_{+}$ be a distance metric on $\mathbb{R}^d$.Then we may impose the following  assumption in Example 1:

\begin{assumption}\label{asm:general.distance}
$\tau(\cdot)$ is a Lipschitz function (with respect to metric $m(\cdot,\cdot)$) with known constant $L$. That is, for any $x,x' \in \mathbb{R}^{d}$, $|\tau(x)-\tau(x')| \leq L m(x,x^\prime)$.
\end{assumption}
For example, $m(x,x^{\prime})=\bigl((x-x^{\prime})^{\top}(x-x^{\prime})\bigr)^{1/2}$ is the Euclidean distance in Example 1 and $m(x,x^{\prime})=\bigl((x-x^{\prime})^{\top}W(x-x^{\prime})\bigr)^{1/2}$ for some positive definite matrix $W$ is a weighted Euclidean distance. Choosing $m(x,x^{\prime})=\left[\tilde{m}(x,x^{\prime})\right]^{\alpha}$ for some $\alpha\in(0,1)$ and some distance measure $\tilde{m}(\cdot,\cdot)$ also allows us to model $\tau$ as  a H\"{o}lder continuous function of order $\alpha$ \citep[][p. 56]{dudley02}. Assumption \ref{asm:general.distance} corresponds to equating $C_{s,s'}$ in Assumption \ref{asm:voltage effect} with $L\cdot m(x_s,x_{s'})$.

In practice, as with matching estimators, distance measures must weight covariates, but there is no obvious best way to do so, and the solution to the $k$-median problem will be sensitive to this choice.  Following the literature \citep[][Section 18.5]{imbens2004nonparametric,imbens2009recent,imbens2015causal}, we think two metrics are reasonable: the Mahalanobis metric, i.e., $m(x,x^{\prime})=\bigl((x-x^{\prime})^{\top}\Sigma_X^{-1}(x-x^{\prime})\bigr)^{1/2}$, where $\Sigma_X$ is the covariance matrix of site-level covariates, or the modification thereof that sets all off-diagonal elements of $\Sigma_X^{-1}$ to zero. Relative to Euclidean distance, the Mahalanobis metric has the advantage of being invariant to changes in the scale of covariates (and more generally, to affine transformations of covariates). However, it is not invariant to other transformations and is not usually recommended for categorical variables. For alternative options, in Section \ref{subsection:structural_model} we also discuss how a structural model for treatment effects can be used naturally to define a metric $m(\cdot,\cdot)$ over covariates.

\subsubsection{Unobserved Treatment Heterogeneity}\label{subsec:unobserved}

Assumption \ref{asm:voltage effect} also accommodates some forms of unobserved treatment heterogeneity. For example,  in our general setup, consider the following assumption:

\begin{assumption}\label{asm:unobserved}
For any $s,s^{\prime} \in \mathcal{S}$, $s\neq s^{\prime}$,
we have
\begin{equation}\label{eq:ass.1.alt}
|\tau_s-\tau_{s^{\prime}}| \leq L\cdot m(X_{s},X_{s^{\prime}})+l, 
\end{equation}
where $L>0$ and $l>0$ are both known, and $m(\cdot,\cdot)$ is a metric. 
\end{assumption}

Assumption \ref{asm:unobserved} allows sites with the same covariates to have different policy effects. The difference, however, is assumed to be at most $l$. It follows that Assumption \ref{asm:unobserved} corresponds to setting $C_{s,s^{\prime}}$ as $L\cdot m(X_{s},X_{s^{\prime}})+l$ in Assumption \ref{asm:voltage effect}.  We motivate Assumption \ref{asm:unobserved} using a simple linear regression model in Appendix  \ref{sec:app.unobserved}.

\subsubsection{Treatment Heterogeneity Implied by Structural Model} \label{subsection:structural_model} 
Our approach can also accommodate structural models of treatment heterogeneity. To see this, let $\tau(X_s)=g(\theta,X_s)$ be the treatment effect of site $s$, built from a structural model $g(\cdot,\cdot)$ with parameters $\theta\in\Theta$ \citep{gechter2023site}. An example of treatment effects that take the form $g(\theta, X)$ for a structural parameter $\theta$ appears in \citeauthor{list2026economics}'s (\citeyear{list2026economics}) discussion of scaling early childhood programs. For instance, if treatment effects on child outcomes depend on site characteristics through a structural model of parental investment and skill production, that structure could discipline which sites are truly close for purposes of extrapolation.

  The site-level treatment effect is then governed by $\tau:=\left\{ \tau(\cdotp)\mid\tau(\cdotp)=g(\theta,\cdotp),\theta\in\Theta\right\}$.
We may define the  pseudometric
\[m(x,x^\prime):=\sup_{\theta\in\Theta}\left|g(\theta,x)-g(\theta,x^{\prime})\right|,\]
implying that $\tau(\cdot)$ informed by the structural model is $1$-Lipschitz continuous with respect to the redefined $m(\cdot,\cdot)$. Moreover, suppose we have a plausible subset $\Theta^* \subset\Theta$ derived from previous studies (for example, a 95\% confidence set for  the true parameter). One may incorporate such information by analogously defining
\[m(x,x^\prime):=\sup_{\theta\in\Theta^*}\left|g(\theta,x)-g(\theta,x^{\prime})\right|.\]
If the true structural paramter is known, then $\Theta^*$ can be taken as a singleton.\footnote{\label{footnote:structural}We note that by considering all $1$-Lipschitz continuous functions with respect to the induced pseudometric $m$, the parameter space is expanded and may contain  functions that are not consistent with the structural model. In this case, our main results may still be viewed as providing an upper bound to the problem in the original parameter space. }

\subsubsection{General Convex and Centrosymmetric Class}\label{subsec:cc}
In Example 1, instead of imposing Assumption \ref{asm:general.distance}, suppose $\tau(\cdotp)\in L_{cc}$, where
$L_{cc}$ is a convex and centrosymmetric (i.e., $\tau(\cdotp)\in L_{cc}$
implies $-\tau(\cdotp)\in L_{cc}$) parameter space. It follows that this also fits Assumption \ref{asm:voltage effect} by setting  

\[
C_{s,s^{\prime}}= \left\{\sup_{\tau(\cdot)\in L_{cc}} \mid\tau(X_s)-\tau(X_{s^\prime})\mid \right\},
\]
which we assume to be finite for all $s,s^{\prime}\in\mathcal{S}$. Analogous to Footnote \ref{footnote:structural}, this construction may allow for additional functions of $\tau(\cdot)\notin L_{cc}$, in which case our results still provide an upper bound to the problem in the original space.

\section{Conclusion}
\label{sec:conclusion}
This paper presented a decision-theoretic justification for viewing the question of how to best choose \emph{where} to experiment in order to optimize external validity as a \emph{$k$-median  problem}. More concretely, we presented  conditions under which minimizing the worst-case, welfare-based regret among all purposive (nonrandomized) schemes that select $k$ sites is equivalent to solving a $k$-median problem.

We believe there are many interesting directions for future work. For example, it would be interesting to better understand the value of randomized sampling schemes and whether site-level covariates can be used to design such randomized selection with an eye on external validity. We also think that discussions around the relation between the $k$-median problem and synthetic purposive sampling of \cite{egamidesigning} open interesting research directions to provide a decision-theoretic justification for the use synthetic control of \cite{abadie2010synthetic}.

\appendix

\bibliographystyle{ecta}
\bibliography{references.bib}

\newpage

\section{Proofs of Main Results}\label{sec:app.main}

\subsection{Proof of Theorem \ref{thm:main}} \label{section:proof_Theorem1}

\subsubsection{Upper Bound} \label{subsection:proof_Lemma1}
We first show that, for any $\mathscr{S} \in \mathcal{A}(k)$, we have
\begin{equation} \label{eqn:upper_bound_thm1}
 \inf_{T \in \mathcal{T}_{\mathscr{S}}} \sup_{\tau \in \tau(C)}  \mathcal{R}(T,\mathscr{S},\tau)  \leq  
\frac{1}{2}\frac{1}{\#\mathcal{S}_{P}}\sum_{s\in\mathcal{S}_{P}}C_{N_{\mathscr{S}}(s),s}.
\end{equation}
In order to show this, fix the selected sites $\mathscr{S}$. For each policy-relevant site $s \in \mathcal{S}_{P}$, recall that its nearest neighbor $N_{\mathscr{S}}(s)$ among selected sites $\mathscr{S}$ satisfies
\[C_{N_{\mathscr{S}}(s),s}\leq C_{s',s}\;\:\forall s'\in\mathscr{S}.\]
If multiple sites satisfy the above relation, $N_{\mathscr{S}}(s)$ is understood to be the one with the smallest index. 

In a slight abuse of notation, let $\widehat{\tau}_{N_{\mathscr{S}}(s)}$ be the estimated treatment effect obtained at experimental site $N_{\mathscr{S}}(s)$. Recall that, from the statistical model in \eqref{eqn:stat_model}, the scalar random variable $\widehat{\tau}_{N_{\mathscr{S}}(s)}$ follows a normal distribution with mean $\tau_{N_{\mathscr{S}}(s)}$ and known standard deviation $\sigma_{N_{{\mathscr{S}}(s)}}>0$. 

For any decision rule $T^* \in \mathcal{T}_{\mathscr{S}}$, the definitions of infimum and supremum yield
\begin{eqnarray*}
\inf_{T \in \mathcal{T}_{\mathscr{S}}} \sup_{\tau \in \tau(C)}  \mathcal{R}(T,\mathscr{S},\tau)  &\leq& \sup_{\tau \in \tau(C)}  \mathcal{R}(T^*,\mathscr{S},\tau),  \\
&\leq & \frac{1}{\# \mathcal{S}_{P}} \sum_{s \in \mathcal{S}_{P}} \sup_{\tau \in \tau(C)} \left( \tau_s \left( \mathbf{1}\{ \tau_s\geq 0\} - \mathbb{E}_{\tau_{\mathscr{S}}}[T^*_{s}(\widehat{\tau}_{\mathscr{S}})] \right) \right).
\end{eqnarray*}

Now, for each policy-relevant site $s \in \mathcal{S}_{P}$, consider the following decision rule that depends on $\widehat{\tau}_{\mathscr{S}}$ only through the estimated treatment effect obtained at experimental site $N_{\mathscr{S}}(s)$: 

\[T^*_{s} \left( \widehat{\tau}_{\mathscr{S}} \right) := \Phi \left( \frac{\widehat{\tau}_{N_{\mathscr{S}}(s)}}{\tilde{\sigma}_{s}} \right), \]
where
\[ \tilde{\sigma}_{s} : = \sqrt{\left(\frac{C_{N_\mathscr{S}(s),s}}{\sqrt{\pi/2}}\right)^2-\sigma^2_{N_{\mathscr{S}}(s)} }.\]

Define $C(\mathscr{S}):=\sqrt{\pi/2}\max_{s\in\mathcal{S}_{P}}\left(\sigma_{N_{\mathscr{S}}(s)}\right)$. Note that when $C_{N_{\mathscr{S}(s)},s}>C(\mathscr{S})$ for all $s\in\mathcal{S}_P$, our suggested decision rule is well defined. 

Moreover, since the decision rule $T_s^*$ depends on $\widehat{\tau}_{\mathscr{S}}$ only via the scalar normally distributed random variable  $\widehat{\tau}_{N_{\mathscr{S}}(s)}$ with mean $\tau_{N_{\mathscr{S}}(s)}$ and known variance $\sigma^2_{N_{\mathscr{S}}(s)}$ we have
\begin{align*}
&\sup_{\tau \in \tau(C)} \left( \tau_s \left( \mathbf{1}\{ \tau_s \geq 0\} - \mathbb{E}_{\tau_{\mathscr{S}}}[T^*_{s}(\widehat{\tau}_{\mathscr{S}})] \right) \right) \\
=&\sup_{\tau \in \tau(C)} \left( \tau_s \left( \mathbf{1}\{ \tau_s \geq 0\} - \mathbb{E}_{\tau_{N_{\mathscr{S}}(s)}}\left[\Phi \left( \frac{\widehat{\tau}_{N_{\mathscr{S}}(s)}}{\tilde{\sigma}_{s}} \right)\right] \right) \right).
\end{align*}

Assumption \ref{asm:voltage effect} implies that if $\tau_s$ and $\tau_{N_{\mathscr{S}}(s)}$ are the $s$ and $N_{\mathscr{S}}(s)$ entries of a vector $\tau \in \tau(C)$, respectively, then  
\[
\tau_s \in I(\tau_{N_{\mathcal{\mathscr{S}}}(s)}):=[\tau_{N_{\mathscr{S}}(s)}-C_{{N_{\mathscr{S}}(s)},s},\tau_{N_{\mathscr{S}}(s)}+C_{{N_{\mathscr{S}}(s)},s}]. 
\]
Thus, 
\begin{eqnarray*}
&&\sup_{\tau \in \tau(C)} \left( \tau_s \left( \mathbf{1}\{ \tau_s \geq 0\} - \mathbb{E}_{\tau_{N_{\mathscr{S}}(s)}}\left[\Phi \left( \frac{\widehat{\tau}_{N_{\mathscr{S}}(s)}}{\tilde{\sigma}_{s}} \right)\right] \right) \right)  \\
&\leq& \sup_{\tau_s \in I(\tau_{N_{\mathcal{\mathscr{S}}}(s)}),\tau_{N_{\mathcal{\mathscr{S}}}(s)}\in\mathbb{R}} \left( \tau_s \left( \mathbf{1}\{ \tau_s \geq 0\} - \mathbb{E}_{\tau_{N_{\mathscr{S}}(s)}}\left[\Phi \left( \frac{\widehat{\tau}_{N_{\mathscr{S}}(s)}}{\tilde{\sigma}_{s}} \right)\right] \right) \right).
\end{eqnarray*}

It has been shown (see, e.g.,\citet[][Proposition 7 and Corollary 8]{stoye2012minimax}, \citet[][Lemma C.4]{olea2023decision}) that whenever $C_{N_{\mathscr{S}(s)},s}>C(\mathscr{S})=\sqrt{\pi/2}\max_{s\in\mathcal{S}_{P}}\left(\sigma_{N_{\mathscr{S}}(s)}\right)\geq\sqrt{\pi/2}\sigma_{N_{\mathscr{S}}(s)}$, the following holds:
\begin{equation} \label{eqn:aux2_lemma1}
\sup_{\tau_s \in I(\tau_{N_{\mathcal{\mathscr{S}}}(s)}),\tau_{N_{\mathcal{\mathscr{S}}}(s)}\in\mathbb{R}} \left( \tau_s \left( \mathbf{1}\{ \tau_s \geq 0\} - \mathbb{E}_{\tau_{N_{\mathscr{S}}(s)}}\left[\Phi \left( \frac{\widehat{\tau}_{N_{\mathscr{S}}(s)}}{\tilde{\sigma}_{s}} \right)\right] \right) \right)= \frac{C_{N_{\mathscr{S}(s)},s}}{2},\forall s\in\mathcal{S}_P,
\end{equation}
We thus conclude that the bound in \eqref{eqn:upper_bound_thm1} holds.

\subsubsection{Lower Bound} \label{subsection:proof_Lemma2}
We now show that, for any $\mathscr{S} \in \mathcal{A}(k)$ we, have
\begin{equation} \label{eqn:lower_bound_thm1}
 \inf_{T \in \mathcal{T}_{\mathscr{S}}} \sup_{\tau \in \tau(C)}  \mathcal{R}(T,\mathscr{S},\tau)  \geq  
\frac{1}{2}\frac{1}{\#\mathcal{S}_{P}}\sum_{s\in\mathcal{S}_{P}}C_{N_{\mathscr{S}}(s),s}.
\end{equation}

The proof has three parts.

\textsc{Part I:} For given $C$ and given selection of experimental sites $\mathscr{S}$, the following special vector of treatment effects $\tau^*$ belongs to the set $\tau(C)$: 
\begin{align*}
\tau^*_{j} & =0,\text{for all }j\in\mathscr{S},\\
\tau^*_{j} & =C_{N_{\mathscr{S}}(j),j},\text{for all }j\in \mathcal{S}\setminus\mathscr{S}.
\end{align*}
Note that $\tau^*$ sets the treatment effect at each experimental site in $\mathscr{S}$ to be equal to zero and the treatment effect at any other site $j$ to the bound on the voltage effect associated with the nearest neighbor $N_{\mathscr{S}}(j)$. To see that $\tau^{*}\in\tau(C)$, consider three cases:

\noindent \emph{Case 1:} Suppose that $s,s' \in \mathscr{S}$. In this case, we trivially have $|\tau_s- \tau_{s'}| = 0 \leq C_{s,s'}$.

\noindent \emph{Case 2:} Suppose that $s \notin \mathscr{S}$, but $s' \in \mathscr{S}$. In this case, 
\[ |\tau_{s'}-\tau_s | =  |  0-C_{N_{\mathscr{S}}(s),s}| = C_{N_{\mathscr{S}(s),s}} \leq C_{s',s}, \]
where the last inequality follows by the definition of $N_{\mathscr{S}(s)}$ and the fact that $s' \in\mathscr{S}$.

\noindent \emph{Case 3:} Finally, suppose $s,s' \notin \mathscr{S}$. In this case, 
\[
\left|\tau_{s}-\tau_{s'}\right|=\left|C_{N_{\mathscr{S}}(s),s}-C_{N_{\mathscr{S}}(s'),s'}\right|.
\]
Without loss of generality, assume that $s \neq s'$ and $C_{N_{\mathscr{S}}(s),s}\geq C_{N_{\mathscr{S}}(s'),s'}$. Then
\begin{align*}
\left|C_{N_{\mathscr{S}}(s),s}-C_{N_{\mathscr{S}}(s'),s'}\right|&=C_{N_{\mathscr{S}}(s),s}-C_{N_{\mathscr{S}}(s'),s'}\\
&\leq C_{N_{\mathscr{S}}(s'),s}-C_{N_{\mathscr{S}}(s'),s'}\\
&\leq C_{N_{\mathscr{S}}(s'),s'}+C_{s,s'} -C_{N_{\mathscr{S}}(s'),s'}\\
&=C_{s',s},
\end{align*}
where the first inequality uses definition of $N_{\mathscr{S}(s)}$, the second inequality uses Assumption \ref{asm:positive.c}, and the last equality uses the fact that $C$ is symmetric. We conclude that $\tau^{*}\in\tau(C)$.

\textsc{Part II:} As $\tau^{*}\in\tau(C)$ and we are agnostic about the sign of the voltage effects, we have $-\tau^{*}\in\tau(C)$. Consider a distribution over $\tau$ such that $\tau=\tau^{*}$
with probability $1/2$ and $\tau=-\tau^{*}$ with probability $1/2$.
Then, for any treatment rule $T$: 
\begin{equation}
\sup_{\tau\in\tau(C)}\mathcal{R}(T,\mathscr{S},\tau)\geq\frac{1}{2}\mathcal{R}(T,\mathscr{S},\tau^{*})+\frac{1}{2}\mathcal{R}(T,\mathscr{S},-\tau^{*}).\label{eqn:Lemma2_aux1}
\end{equation}
Moreover, by definition, for all $\tau$, 
\[
\mathcal{R}(T,\mathscr{S},\tau):=\frac{1}{\#\mathcal{S}_{P}}\sum_{s\in\mathcal{S}_{P}}\tau_{s}\left(\mathbf{1}\{\tau_{s}\geq0\}-\mathbb{E}_{\tau_{\mathscr{S}}}[T_{s}(\widehat{\tau}_{\mathscr{S}})]\right).
\]
Therefore, 
\begin{eqnarray*}
\mathcal{R}(T,\mathscr{S},\tau^{*}) & = & \frac{1}{\#\mathcal{S}_{P}}\sum_{s\in\mathcal{S}_{P}}\tau_{s}^{*}\left(\mathbf{1}\{\tau_{s}^{*}\geq0\}-\mathbb{E}_{\tau_{\mathscr{S}}^{*}}[T_{s}(\widehat{\tau}_{\mathscr{S}})]\right)\\
 & = & \frac{1}{\#\mathcal{S}_{P}}\sum_{s\in\mathcal{S}_{P}}C_{N_{\mathscr{S}}(s),s}\left(1-\mathbb{E}_{\mathbf{0}}[T_{s}(\widehat{\tau}_{\mathscr{S}})]\right),
\end{eqnarray*}
where the last line uses the definition of $\tau^{*}$ and the assumption that  $C_{N_{\mathscr{S}}(s),s}>0$ for all $s\in\mathcal{S}_P$.
As $-\tau_{s}^{*}<0$ for all $s\in\mathcal{S}_{P}$, 
we analogously have 
\begin{eqnarray*}
\mathcal{R}(T,\mathscr{S},-\tau^{*}) & =\frac{1}{\#\mathcal{S}_{P}}\sum_{s\in\mathcal{S}_{P}}C_{N_{\mathscr{S}}(s),s}\mathbb{E}_{\mathbf{0}}[T_{s}(\widehat{\tau}_{\mathscr{S}})].
\end{eqnarray*}
Conclude that for any $T\in\mathcal{T}_{\mathscr{S}}$,
\begin{align}
\sup_{\tau\in\tau(C)}\mathcal{R}(T,\mathscr{S},\tau) & \geq\frac{1}{2}\frac{1}{\#\mathcal{S}_{P}}\sum_{s\in\mathcal{S}_{P}}C_{N_{\mathscr{S}}(s),s}\left(1-\mathbb{E}_{\mathbf{0}}[T_{s}(\widehat{\tau}_{\mathscr{S}})]\right)\nonumber \\
 & +\frac{1}{2}\frac{1}{\#\mathcal{S}_{P}}\sum_{s\in\mathcal{S}_{P}}C_{N_{\mathscr{S}}(s),s}\mathbb{E}_{\mathbf{0}}[T_{s}(\widehat{\tau}_{\mathscr{S}})]\nonumber \\
 & =\frac{1}{2}\frac{1}{\#\mathcal{S}_{P}}\sum_{s\in\mathcal{S}_{P}}C_{N_{\mathscr{S}}(s),s}\left\{ 1-\mathbb{E}_{\mathbf{0}}[T_{s}(\widehat{\tau}_{\mathscr{S}})]+\mathbb{E}_{\mathbf{0}}[T_{s}(\widehat{\tau}_{\mathscr{S}})]\right\} \nonumber \\
 & =\frac{1}{2}\frac{1}{\#\mathcal{S}_{P}}\sum_{s\in\mathcal{S}_{P}}C_{N_{\mathscr{S}}(s),s}.\label{eqn:Lemma2_aux3}
\end{align}

\textsc{Part III:} Equation (\ref{eqn:Lemma2_aux3}) implies that the desired lower bound in \eqref{eqn:lower_bound_thm1} holds.

\subsection{Proof of Theorem \ref{thm:no.random}} \label{section:proof_no_random}
When randomized site selection is allowed, a decision maker chooses
(i) a probability vector $p\in\Delta(\mathcal{A}(k))$ so that each
element of $p$, denoted as $p(\mathscr{S})$, represents the probability
of selecting $\mathscr{S}$, and (ii) a treatment rule $T_{\mathscr{S}}\in\mathcal{T}_{\mathscr{S}}$
for all $\mathscr{S}\in\mathcal{A}(k)$. Thus, the value of the MMR
site selection problem becomes
\[
V^{\star}=\inf_{p\in\Delta(\mathcal{A}(k)),T_{\mathscr{S}}\in\mathcal{T}_{\mathscr{S}},\mathscr{S}\in\mathcal{A}(k)}\sup_{\tau\in\tau(C)}\sum_{\mathscr{S}\in\mathcal{A}(k)}p(\mathscr{S})\cdot\mathcal{R}(T_{\mathscr{S}},\mathscr{S},\tau).
\]
For each $s\in\mathcal{S}_{P}$, denote
by $N_{\mathcal{S}_{E}}(s)\in\mathcal{S}_{E}$ a ``nearest neighbor'' of site $s\in\mathcal{S}_{P}$
among all sites in $\mathcal{S}_{E}$, i.e., 
\[
C_{N_{\mathcal{S}_{E}}(s),s}\leq C_{s',s}\;\:\forall s'\in\mathcal{S}_{E}.
\] The proof is split into two steps.

\textsc{Step 1}: We show $V^{\star}\geq\frac{1}{2}\frac{1}{\#\mathcal{S}_{P}}\sum_{s\in\mathcal{S}_{P}}C_{N_{\mathcal{S}_{E}}(s),s}$.
Consider the following specification for the treatment effects: 
\begin{align*}
\tau_{j}^{\star} & =0,\text{ for all }j\in\mathcal{S}_{E},\\
\tau_{j}^{\star} & =C_{N_{\mathcal{S}_{E}}(j),j},\text{ for all }j\in\mathcal{S}_{P}.
\end{align*}
By analogous arguments to Part I of Section \ref{subsection:proof_Lemma2}, we can verify that $\tau^{\star}\in\tau(C)$ and $-\tau^{\star}\in\tau(C)$.
Suppose Nature randomizes evenly over $\{\tau^{\star},-\tau^{\star}\}$. Then for each $p\in\Delta(\mathcal{A}(k))$
and $T_{\mathscr{S}}\in\mathcal{T}_{\mathscr{S}}$, $\mathscr{S}\in\mathcal{A}(k)$, we have
\begin{align*}
 & \sup_{\tau\in\tau(C)}\sum_{\mathscr{S}\in\mathcal{A}(k)}p(\mathscr{S})\cdot\mathcal{R}(T_{\mathscr{S}},\mathscr{S},\tau)\\
\geq & \frac{1}{2}\sum_{\mathscr{S}\in\mathcal{A}(k)}p(\mathscr{S})\mathcal{R}(T_{\mathscr{S}},\mathscr{S},\tau^{\star})+\frac{1}{2}\sum_{\mathscr{S}\in\mathcal{A}(k)}p(\mathscr{S})\mathcal{R}(T_{\mathscr{S}},\mathscr{S},-\tau^{\star})\\
= & \frac{1}{2}\sum_{\mathscr{S}\in\mathcal{A}(k)}p(\mathscr{S})\left(\mathcal{R}(T_{\mathscr{S}},\mathscr{S},\tau^{\star})+\mathcal{R}(T_{\mathscr{S}},\mathscr{S},-\tau^{\star})\right).
\end{align*}
Furthermore, 
\begin{align*}
\mathcal{R}(T_{\mathscr{S}},\mathscr{S},\tau^{\star}) & =\frac{1}{\#\mathcal{S}_{P}}\sum_{s\in\mathcal{S}_{P}}\tau_{s}^{\star}\left(\mathbf{1}\{\tau_{s}^{\star}\geq0\}-\mathbb{E}_{\tau_{\mathscr{S}}^{\star}}\left[\left(T_{\mathscr{S}}\right)_{s}(\widehat{\tau}_{\mathscr{S}})\right]\right)\\
 & =\frac{1}{\#\mathcal{S}_{P}}\sum_{s\in\mathcal{S}_{P}}C_{N_{\mathcal{S}_{E}}(s),s}\left(1-\mathbb{E}_{\mathbf{0}}\left[\left(T_{\mathscr{S}}\right)_{s}(\widehat{\tau}_{\mathscr{S}})\right]\right)
\end{align*}
and analogously, 
\begin{align*}
\mathcal{R}(T_{\mathscr{S}},\mathscr{S},-\tau^{\star}) & =\frac{1}{\#\mathcal{S}_{P}}\sum_{s\in\mathcal{S}_{P}}C_{N_{\mathcal{S}_{E}}(s),s}\mathbb{E}_{\mathbf{0}}[T_{s}(\widehat{\tau}_{\mathscr{S}})].
\end{align*}
Conclude that, for any $\mathscr{S}\in\mathcal{A}(k)$ and $T_{\mathscr{S}}\in \mathcal{T}_\mathscr{S}$,
we have 
\begin{align*}
\mathcal{R}(T_{\mathscr{S}},\mathscr{S},\tau^{\star})+\mathcal{R}(T_{\mathscr{S}},\mathscr{S},-\tau^{\star})=\frac{1}{\#\mathcal{S}_{P}}\sum_{s\in\mathcal{S}_{P}}C_{N_{\mathcal{S}_{E}}(s),s},
\end{align*}
implying that
\begin{align*}
 & \sum_{\mathscr{S}\in\mathcal{A}(k)}p(\mathscr{S})\left(\mathcal{R}(T_{\mathscr{S}},\mathscr{S},\tau^{\star})+\mathcal{R}(T_{\mathscr{S}},\mathscr{S},-\tau^{\star})\right)\\
= & \sum_{\mathscr{S}\in\mathcal{A}(k)}\left[p(\mathscr{S})\left(\frac{1}{\#\mathcal{S}_{P}}\sum_{s\in\mathcal{S}_{P}}C_{N_{\mathcal{S}_{E}}(s),s}\right)\right]\\
= & \frac{1}{\#\mathcal{S}_{P}}\sum_{s\in\mathcal{S}_{P}}C_{N_{\mathcal{S}_{E}}(s),s}
\end{align*}
as $\sum_{\mathscr{S}\in\mathcal{A}(k)}p(\mathscr{S})=1$ and $\frac{1}{\#\mathcal{S}_{P}}\sum_{s\in\mathcal{S}_{P}}C_{N_{\mathcal{S}_{E}}(s),s}$
does not depend on $p(\mathscr{S})$. We conclude that $V^{\star}\geq \frac{1}{2}\frac{1}{\#\mathcal{S}_{P}}\sum_{s\in\mathcal{S}_{P}}C_{N_{\mathcal{S}_{E}}(s),s}$
as required. 

\textsc{Step 2}: We show $V^{\star}\leq\frac{1}{2}\frac{1}{\#\mathcal{S}_{P}}\sum_{s\in\mathcal{S}_{P}}C_{N_{\mathcal{S}_{E}}(s),s}$. Denote by $\mathscr{S}^{\star}$ a solution of the $k$-median problem
\eqref{eqn:k-median}. Under the conditions of Theorem \ref{thm:main}, we have
\[
\inf_{T\in\mathcal{T}_{\mathscr{\mathscr{S}^{\star}}}}\sup_{\tau\in\tau(C)}\mathcal{R}(T,\mathscr{S}^{\star},\tau)=\frac{1}{2}\frac{1}{\#\mathcal{S}_{P}}\sum_{s\in\mathcal{S}_{P}}C_{N_{\mathscr{\mathscr{S}^{\star}}}(s),s}.
\]
By definition, $V^{\star}\leq\inf_{T\in\mathcal{T}_{\mathscr{\mathscr{S}^{\star}}}}\sup_{\tau\in\tau(C)}\mathcal{R}(T,\mathscr{S}^{\star},\tau)$.
Thus, it suffices to show that
\[
\frac{1}{2}\frac{1}{\#\mathcal{S}_{P}}\sum_{s\in\mathcal{S}_{P}}C_{N_{\mathscr{\mathscr{S}^{\star}}}(s),s}\leq\frac{1}{2}\frac{1}{\#\mathcal{S}_{P}}\sum_{s\in\mathcal{S}_{P}}C_{N_{\mathcal{S}_{E}}(s),s}.
\]
To this end, note that $\mathcal{\mathcal{N}}\in\mathcal{A}(k)$ and that
the $k$-median problem solves 
\[
\inf_{\mathscr{S}\in\mathcal{A}(k)}\frac{1}{2}\frac{1}{\#\mathcal{S}_{P}}\sum_{s\in\mathcal{S}_{P}}C_{N_{\mathscr{S}}(s),s}.
\]
So, it must be the case that 
\begin{align*}
\frac{1}{2}\frac{1}{\#\mathcal{S}_{P}}\sum_{s\in\mathcal{S}_{P}}C_{N_{\mathscr{\mathscr{S}^{\star}}}(s),s} & \leq\frac{1}{2}\frac{1}{\#\mathcal{S}_{P}}\sum_{s\in\mathcal{S}_{P}}C_{N_{\mathscr{\mathcal{N}}}(s),s}=\frac{1}{2}\frac{1}{\#\mathcal{S}_{P}}\sum_{s\in\mathcal{S}_{P}}C_{N_{\mathcal{S}_{E}}(s),s},
\end{align*}
where the equality follows because, by our assumption, $\mathcal{N}$ is also such that $\min_{s' \in \mathcal{N}}C_{s',s} =\min_{s' \in \mathcal{S}_E}C_{s',s}$,
$\forall s\in\mathcal{S}_{P}$, implying $C_{N_{\mathscr{\mathcal{N}}}(s),s}=C_{N_{\mathcal{S}_{E}}(s),s}$
for all $s\in\mathcal{S}_{P}$. 

\textsc{Conclusion}: Based on steps 1 and 2 above, we conclude that under conditions of Theorem \ref{thm:no.random}, we have
$V^{\star}=\frac{1}{2}\frac{1}{\#\mathcal{S}_{P}}\sum_{s\in\mathcal{S}_{P}}C_{N_{\mathcal{S}_{E}}(s),s}$
and it is achieved by the solution of \eqref{eqn:k-median}.

\section{Supplemental Appendix}

\subsection{Gurobi's Output} \label{sec:Gurobi} 

In this section, we discuss some parameters in the \texttt{Gurobi} optimizer that can be tuned to improve its performance.\footnote{For more specifics, see \url{https://www.gurobi.com/documentation/current/refman/mip_models.html}.} First, \texttt{MIPFocus} specifies, on a high level, whether the solution should prioritize speed or optimality. The default value of this parameter is $0$, which achieves a balance between searching for new solutions and proving optimality of the current solution. We set this parameter to $2$, prioritizing finding good-quality and optimal solutions. A second set of important parameters affects how the solver terminates. There are two termination choices for \texttt{MIP} models: 1) a restriction on the runtime of the algorithm, such as using \texttt{TimeLimit} to limit the wall-clock runtime, and 2) controlling the \emph{optimality} gap, by setting a parameter \texttt{MIPGap} that stops the algorithm when the relative gap between the best-known solution and the best known bound on the solution objective is less than the specified value. In this application, we let the runtime be the default value (of infinity) and set the tolerance level to $10^{-4}$. In Figure \ref{fig:gurobi_output}, the \texttt{Gurobi} solver outputs not only the solution and the optimal value of the problem but also the output gap. In this example, we can see that the gap between bounds is $0$, demonstrating that the $k$-median problem has been solved to \emph{provable optimality}. An interesting feature of using the integer programming formulation of the $k$-median problem is that even when we are not able to solve the problem to provable optimality, the optimality gap provides a \emph{suboptimality} guarantee for the obtained solution. Finally, the output reports the algorithm's runtime, which was $0.05$ seconds for this example. There are more parameters one can tune to control the root relaxation, the aggressiveness of the cutting plane strategies and the level of presolve, tolerance parameters for primal feasibility, the integer feasibility, and more. For the applications in this paper, we let the remaining parameters take their default values.

\begin{figure}[!ht]
    \begin{center}
        \includegraphics[width=0.9\linewidth]{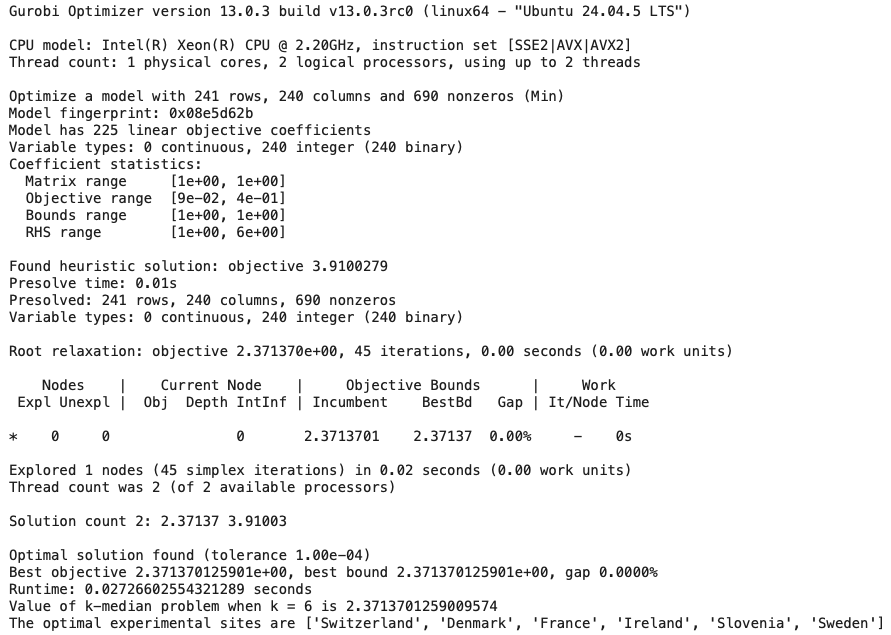}
    \caption{Gurobi Output}
    \label{fig:gurobi_output}
    \end{center}
{\raggedright \footnotesize \textit{Notes}: An example output from the \texttt{MIP} solver in \texttt{Gurobi} using the multi-country experiment for $k = 6$. See details of the application in Section \ref{sec:multicountry}.}
\end{figure}

\subsection{Exact Statement and Proof of Result Discussed in Section \ref{subsec:randomize}} \label{subsection:proof_randomize}
Consider the following specialization of our setting: $\mathcal{S}_E=\mathcal{S}_P=\{1,2\}$ (because there are no voltage effects, the notational distinction between $\mathcal{S}_P$ and $\mathcal{S}_E$ can be dropped); we know $|\tau_1-\tau_2|\leq 1$; and sampling distributions of signals are degenerate, i.e., $\widehat{\tau}_s=\tau_s$. We then have:\footnote{This closed-form solution was discovered using ChatGPT Pro.}
\begin{lem}\label{lem:app1}
    Under this subsection's additional assumptions as just stated:
    \begin{enumerate}
        \item If sampling schemes may randomize, the lowest expected regret achievable by any combination of sampling scheme and treatment rule equals $3/20$.
        \item If sampling must be purposive (nonrandomized), this value increases to $1/4$.
    \end{enumerate}
\end{lem}
\begin{proof} 
We begin with the second claim. By symmetry, it is w.l.o.g. to suppose that site $1$ will be sampled. The policy maker will learn that site's treatment effect. The decision problem is therefore nontrivial only with respect to site $2$, whose treatment effect is learned to lie in $[\tau_1-1,\tau_1+1]$. Optimal treatment assignment for site $2$ is, then, equivalent to the main decision problem in \cite{Manski2007} and is solved by assigning treatment to site $2$ with probability $[(1+\tau_1)/2]_0^1$, where the notation $[X]_0^1 := \min\{\max\{X,0\},1\}$ indicates clamping to $[0,1]$. The value of the site $2$ problem in isolation is $1/2$; since $\tau_1$ is perfectly learned, the overall value of the problem is $1/4$, attained by a least favorable prior that randomizes uniformly over $(\tau_1,\tau_2)=\pm(0,1)$.

Next, MMR is at least $3/20$, the value of the best response to the following prior:
\smallskip

\begin{tabular}{c|cccccc}
$\bm{(\tau_1,\tau_2)=\ldots}$ & $\bigl(\tfrac{1}{2},-\tfrac{1}{2}\bigr)$ & $\bigl(-\tfrac{1}{2},\tfrac{1}{2}\bigr)$ & $\bigl(\tfrac{1}{2},\tfrac{3}{2}\bigr)$ & $\bigl(\tfrac{3}{2},\tfrac{1}{2}\bigr)$ & $\bigl(-\tfrac{1}{2},-\tfrac{3}{2}\bigr)$ & $\bigl(-\tfrac{3}{2},-\tfrac{1}{2}\bigr)$ \\
$\bm{\ldots}$\textbf{with probability:} & $.3$ & $.3$ & $.1$ & $.1$ & $.1$ & $.1$   
\end{tabular}
\smallskip

To verify this, suppose that site $1$ was in fact sampled. Then the sign of $\tau_1$ is revealed. The sign of $\tau_2$ is revealed with probability $.2$ (corresponding to the $4^{th}$ and $6^{th}$ possibility above); else, expected regret incurred at site $2$ equals $3/4 \cdot 1/2 = 1/4 \cdot 3/2 = 3/8$ (the expressions correspond to assigning treatment or not, which turns out to be indifferent). Recalling that this occurs with probability $4/5$ and regret is $0$ otherwise and that regret is always zero at site $1$, expected regret is seen to equal $3/8 \cdot 4/5 \cdot 1/2=3/20$. The same calculation applies after sampling site $2$.

Finally, MMR is at most $3/20$, achieved by the following decision rule: First, sample either site $1$ or $2$ with equal probability. Next, if site $1$ was sampled, assign treatment at site $2$ with probability
\begin{equation*}
\pi(\tau_1)=
\begin{cases}
0, & \tau_1\le -1,\\
\displaystyle \frac{1}{2}+\frac{2}{15}\tau_1(1+|\tau_1|),
& -1<\tau_1<1,\\
1, & \tau_1\ge 1
\end{cases}
\end{equation*}
and conversely if site $2$ was sampled. (The decision at the sampled site is trivial.)

Careful book-keeping reveals that worst-case regret of this decision rule equals
\begin{equation*}
\max_{\substack{(\tau_1,\tau_2): \\ \vert \tau_1-\tau_2\vert \leq 1}} \frac{1}{4}~\underset{=:\psi(\tau_1,\tau_2)}{\underbrace{\bigl(\tau_2^+(1-\pi(\tau_1))+\tau_2^- \pi(\tau_1)+\tau_1^+(1-\pi(\tau_2))+\tau_1^- \pi(\tau_2)\bigr)}},
\end{equation*}
where $X^+:=\max\{X,0\}$ and $X^-:=\max\{-X,0\}$. We evaluate this by case-by-case bounding of $\psi(\cdot)$, always using $|\tau_1-\tau_2|\leq 1$.

By symmetry, restrict attention to $\tau_1\geq 0$. If $\tau_1=0$, then $\psi(\tau_1,\tau_2)=|\tau_2|/2 \leq 1/2$ (and similarly for $\tau_2=0$). If $\tau_1>0$ and $\tau_2<0$, then
\begin{eqnarray*}
    \psi(\tau_1,\tau_2) &=&-\tau_2\pi(\tau_1)+\tau_1(1-\pi(\tau_2)) \\
&=& \frac{\tau_1-\tau_2}{2} - \frac{2}{15}\tau_1 \tau_2 (2+\tau_1-\tau_2) \\
   & \leq& \frac{1}{2} + \frac{2}{15} \cdot \frac{1}{4} \cdot 3 = \frac{3}{5},
\end{eqnarray*}
where the last step repeatedly uses $\tau_1-\tau_2\leq 1$ (which implies $-\tau_1\tau_2\leq 1/4$).

If $\tau_1>0$ and $\tau_2> 0$, then
\begin{eqnarray*}
     \psi(\tau_1,\tau_2) = \tau_2(1-\pi(\tau_1))+\tau_1(1-\pi(\tau_2)) .
\end{eqnarray*}
By symmetry, it is w.l.o.g. to suppose $\tau_1 \leq \tau_2 \leq \tau_1+1$. If $\tau_1>1$, $\psi(\tau_1,\tau_2)=0$. If $\tau_2\geq 1$,
$$\psi(\tau_1,\tau_2)=\tau_2(1-\pi(\tau_1)) \leq (\tau_1+1)(1-\pi(\tau_1)).$$
If $\tau_2<1$, monotonicity of $\pi(\cdot)$ implies that
$$\tau_2(1-\pi(\tau_1))+\tau_1(1-\pi(\tau_2)) \leq (\tau_1+\tau_2)(1-\pi(\tau_1)) \leq (\tau_1+1)(1-\pi(\tau_1)).$$
In all subcases, we have
\begin{eqnarray*}
   \psi(\tau_1,\tau_2) &\leq &(\tau_1+1)(1-\pi(\tau_1)) \\
    &=&(\tau_1+1)\biggl(\frac{1}{2}-\frac{2}{15}\tau_1(1+\tau_1)\biggr) \\
    &=& \frac{3}{5}-\frac{(\tau_1+3)(2\tau_1-1)^2}{30},
    \end{eqnarray*}
implying the same bound of $3/5$. 
\end{proof}

\subsection{A Simple Example that Motivates Assumption \ref{asm:unobserved}}\label{sec:app.unobserved}

In this section, we provide a simple linear regression example to motivate Assumption
\ref{asm:unobserved}. Suppose the effect of the status quo is known
in all sites and normalized to zero. For each site $s\in\mathcal{S}$,
we have a random sample of $n_s$ units in the experiment. Let $Y_{i,s}$
be the outcome under the policy of interest for unit $i\in \{1,...,n_s\}$. We
assume that $Y_{i,s}$ is generated as follows:
\[
Y_{i,s}=\beta X_{i,s}+\gamma Z_{i,s}+\varepsilon_{i,s},
\]
where $X_{i,s}\in\mathbb{R}$ is the observed unit-level covariate
for individual $i$ in site $s$, $Z_{i,s}\in\mathbb{R}$ is the \emph{unobserved}
unit-level covariate, $\varepsilon_{i,s}\sim \mathcal{N}(0,\sigma_{\varepsilon,s}^{2})$
is an error term with $\sigma_{\varepsilon,s}^{2}>0$ and is independent of $\left(X_{i,s},Z_{i,s}\right)$,
and $\beta,\gamma\in\mathbb{R}$ are the same across different sites.
For simplicity, suppose that $X_{i,s}$ and $Z_{i,s}$ are jointly
normal: $\left(X_{i,s},Z_{i,s}\right)^{\top}\sim \mathcal{N}(\mu_{s},\varSigma_s)$,
where $\mu_{s}\in\mathbb{R}^{2}$ and $\varSigma_s$ is positive definite. 

Let $\bar{Y}_{s}:=\frac{1}{n_s}\sum_{i=1}^{n_s}Y_{i,s}$ be the sample
average of the observed outcome at site $s$ and $\bar{X}_{s}:=\frac{1}{n_s}\sum_{i=1}^{n_s}X_{i,s},\bar{Z}_{s}:=\frac{1}{n_s}\sum_{i=1}^{n_s}Z_{i,s}$
the observed and \emph{unobserved} site-level aggregate covariates, respectively.
Under the above assumptions, we have
\begin{equation}
\bar{Y}_{s}\mid\bar{X}_{s}\sim \mathcal{N}(\beta\bar{X}_{s}+\gamma\mathbb{E}[\bar{Z}_{s}|\bar{X}_{s}],\sigma_{s}^{2}),\label{eq:app.unobserved.motivation}
\end{equation}
where $\mathbb{E}[\bar{Z}_{s}|\bar{X}_{s}]$ is the expectation of
$\bar{Z}_{s}$ conditional on observed site-level
aggregate covariate $\bar{X}_{s}$ and where $\sigma_{s}^{2}>0$. Then, Assumption \ref{asm:unobserved} models a
case for which the policy effect of interest is
\[
\tau_{s}:=\beta\bar{X}_{s}+\gamma\mathbb{E}[\bar{Z}_{s}|\bar{X}_{s}],\forall s\in\mathcal{S}.
\]
Furthermore, (\ref{eq:app.unobserved.motivation}) implies that, conditional
on $\bar{X}_{s}$, we have an unbiased and normal estimator for $\tau_{s}$.
We may also calculate that $\mathbb{E}[\bar{Z}_{s}|\bar{X}_{s}]=\alpha_s$ for some $\alpha_s$ that depends on each site $s$. 
Then, for $s,s^{\prime}\in\mathcal{S}$, $s\neq s^{\prime}$, we have:
\begin{align}
\tau_{s}-\tau_{s^{\prime}} & =\beta\left(\bar{X}_{s}-\bar{X}_{s^{\prime}}\right)+\gamma\left(\alpha_s-\alpha_{s^{\prime}}\right), \label{eq:app.unobserved.motivation.2}
\end{align}
implying 
that $\left|\tau_{s}-\tau_{s^{\prime}}\right|$ can be bounded as in
Assumption \ref{asm:unobserved} for some positive $L$ and $l$,
as long as we are willing  to assume that $\beta,$ $\gamma$ are bounded, and that 
$\left|\alpha_{s}-\alpha_{s^{\prime}}\right|$ are bounded uniformly among all $s,s^{\prime}\in\mathcal{S}$.

\subsection{Accounting for Voltage Effects Not Sufficiently Large}\label{sec:vol.not.large}
In this section, we discuss how the $k$-median problem can be modified
if bounds on the voltage effects are not
sufficiently large to invoke Theorem \ref{thm:main}. Note the value of the MMR purposive sampling scheme  and treatment rule  can always be upper bounded as
\begin{align}\label{eq:site-by-site-upper-bound}
\inf_{T\in\mathcal{T}_{\mathscr{S}}}\sup_{\tau\in\tau(C)}\mathcal{R}(T,\mathscr{S},\tau)\leq & \frac{1}{\#\mathcal{S}_{P}}\sum_{s\in\mathcal{S}_{P}}\sup_{\tau\in\tau(C)}\left(\tau_{s}\left(\mathbf{1}\{\tau_{s}\geq0\}-\mathbb{E}_{\tau_{\mathscr{S}}}[T_{s}^{*}(\widehat{\tau}_{\mathscr{S}})]\right)\right),
\end{align}
where $T_{s}^{*}(\widehat{\tau}_{\mathscr{S}})$ denotes the optimal treatment rule for $s\in\mathcal{S}_{P}$ when $\mathscr{S}\in\mathcal{A}(k)$
is selected that ignores the cross-site restrictions.  By applying \citet[][Proposition 1]{olea2023decision}, one may still characterize the value of $\sup_{\tau\in\tau(C)}\left(\tau_{s}\left(\mathbf{1}\{\tau_{s}\geq0\}-\mathbb{E}_{\tau_{\mathscr{S}}}[T_{s}^{*}(\widehat{\tau}_{\mathscr{S}})]\right)\right)$ and the form of $T_{s}^{*}(\widehat{\tau}_{\mathscr{S}})$ for all $s\in\mathcal{S}_{P}$, irrespective of magnitude of the voltage effects. In light of this observation, we think it is reasonable to
select sites by minimizing the RHS of \eqref{eq:site-by-site-upper-bound}.

\subsubsection{Results from Brute-Force}

For the sake of exposition, we illustrate this in the context of our second application, in which we have $\#\mathcal{S}_E= \#\mathcal{S}_P  = 15$ so that we set $\mathcal{S}=\{1,...,30\}$, where the first 15 elements of $\mathcal{S}$ are used to index the experimental sites and the rest to index policy relevant sites. The distance matrix is chosen as follows: for any $i \in \mathcal{S}_{E}$ and $j \in \mathcal{S}_{P}$ with $j \neq (15+i) $, we set $C_{i,j}=L\| X_i - X_{j-15} \|$. If $j=15+i$, we set $C_{i,15+i}=L\min_{i' \in \{1,\ldots, 15\} \backslash i }\| X_{i'} - X_i \|(1-\epsilon)$ where we set $\epsilon=0.01$. And we choose $C_{i,i}=0$ for all $i\in\{1,...,30\}$.

For each $\mathscr{S}\in\mathcal{A}(k)$, recall that $\ensuremath{\widehat{\tau}_{\mathscr{S}}}$ refers to the vector of estimates from experimental sites $\mathscr{S}$, and  $\sigma_{N\mathscr{_{S}}(s)}$ refers to the standard deviation
of the nearest neighbor experimental site among $\mathscr{S}$ for policy relevant site
$s\in\mathcal{S}_P$. Without loss of generality, for each $s$, we reorder $\ensuremath{\widehat{\tau}_{\mathscr{S}}}$
so that the first element of $\ensuremath{\widehat{\tau}_{\mathscr{S}}}$
is the nearest neighbor of $s$ and is indexed as $N\mathscr{_{S}}(s)$,
the second element is the second closest and indexed as $N^{2}\mathscr{_{S}}(s)$,
etc. Note this order may be different for each site $s$, and  our construction of the distance matrix $C$ guarantees  that, for this dataset, $N_{\mathscr{S}}(s)$ is unique for each $\mathscr{S}$ and $s$.  We refer to  $k(\mathscr{S})$ as the cardinality of $\mathscr{S}$. By \citet[][Proposition 1]{olea2023decision}, we have

\[
T_{s}^{*}\left(\widehat{\tau}_{\mathscr{S}}\right):=\begin{cases}
\mathbf{1}\left\{ w_{s}^{\top}\widehat{\tau}_{\mathscr{S}}\geq0\right\} , & \text{if }C_{N\mathscr{_{S}}(s),s}<\sqrt{\frac{\pi}{2}}\sigma_{N\mathscr{_{S}}(s)},\\

\mathbf{1}\{\widehat{\tau}_{N_{\mathscr{S}}(s)}\geq 0\}, & \text{if }C_{N\mathscr{_{S}}(s),s}=\sqrt{\frac{\pi}{2}}\sigma_{N\mathscr{_{S}}(s)},\\

\Phi\left(\frac{\widehat{\tau}_{N\mathscr{_{S}}(s)}}{\sqrt{\left(\frac{C_{N\mathscr{_{S}}(s),s}}{\sqrt{\pi/2}}\right)^{2}-\sigma_{N\mathscr{_{S}}(s)}^{2}}}\right), & \text{if }C_{N\mathscr{_{S}}(s),s}>\sqrt{\frac{\pi}{2}}\sigma_{N\mathscr{_{S}}(s)},
\end{cases}
\]
where 
\[
w_{s}^{\top}:=\left(1,\frac{\max\{m_{s}^{*}-C_{N^{2}\mathscr{_{S}}(s),s},0\}/\sigma_{N^{2}\mathscr{_{S}}(s)}^{2}}{(m_{s}^{*}-C_{N\mathscr{_{S}}(s),s})/\sigma_{N\mathscr{_{S}}(s)}^{2}},\ldots,\frac{\max\{m_{s}^{*}-C_{N^{k(\mathscr{S})}\mathscr{_{S}}(s),s},0\}/\sigma_{N^{k(\mathscr{S})}\mathscr{_{S}}(s)}^{2}}{(m_{s}^{*}-C_{N\mathscr{_{S}}(s),s})/\sigma_{N\mathscr{_{S}}(s)}^{2}}\right),
\]
and $m_{s}^{*}\geq C_{N\mathscr{_{S}}(s),s}$ is uniquely  determined by a simple fixed point program for each $s\in\mathcal{S}_P$. It follows that the RHS of \eqref{eq:site-by-site-upper-bound} equals 
\begin{align*}
& \frac{1}{\#\mathcal{S}_{P}}\sum_{s\in\mathcal{S}_{P}:C_{N\mathscr{_{S}}(s),s}>\sqrt{\frac{\pi}{2}}\sigma_{N\mathscr{_{S}}(s)}}\frac{C_{N\mathscr{_{S}}(s),s}}{2}\\
+ & \frac{1}{\#\mathcal{S}_{P}}\sum_{s\in\mathcal{S}_{P}:C_{N\mathscr{_{S}}(s),s}\leq\sqrt{\frac{\pi}{2}}\sigma_{N\mathscr{_{S}}(s)}}m_{s}^{*}\Phi\left(-\sqrt{\sum_{j=1}^{k(\mathscr{S})}\frac{\max^{2}\left\{ m_{s}^{*}-C_{N^{j}\mathscr{_{S}}(s),s},0\right\} }{\sigma_{N^{j}\mathscr{_{S}}(s)}^{2}}}\right).
\end{align*}
One can always minimize the above upper bound among all $\mathscr{S}\in \mathcal{A}(k)$ by brute force. Since we have knowledge of the variance of the estimate from each experimental site from the second application, we report the corresponding brute-force results  in Figure \ref{fig:small.c}. 

\begin{figure}[!ht]
    \centering
    \includegraphics[width=.7\linewidth]{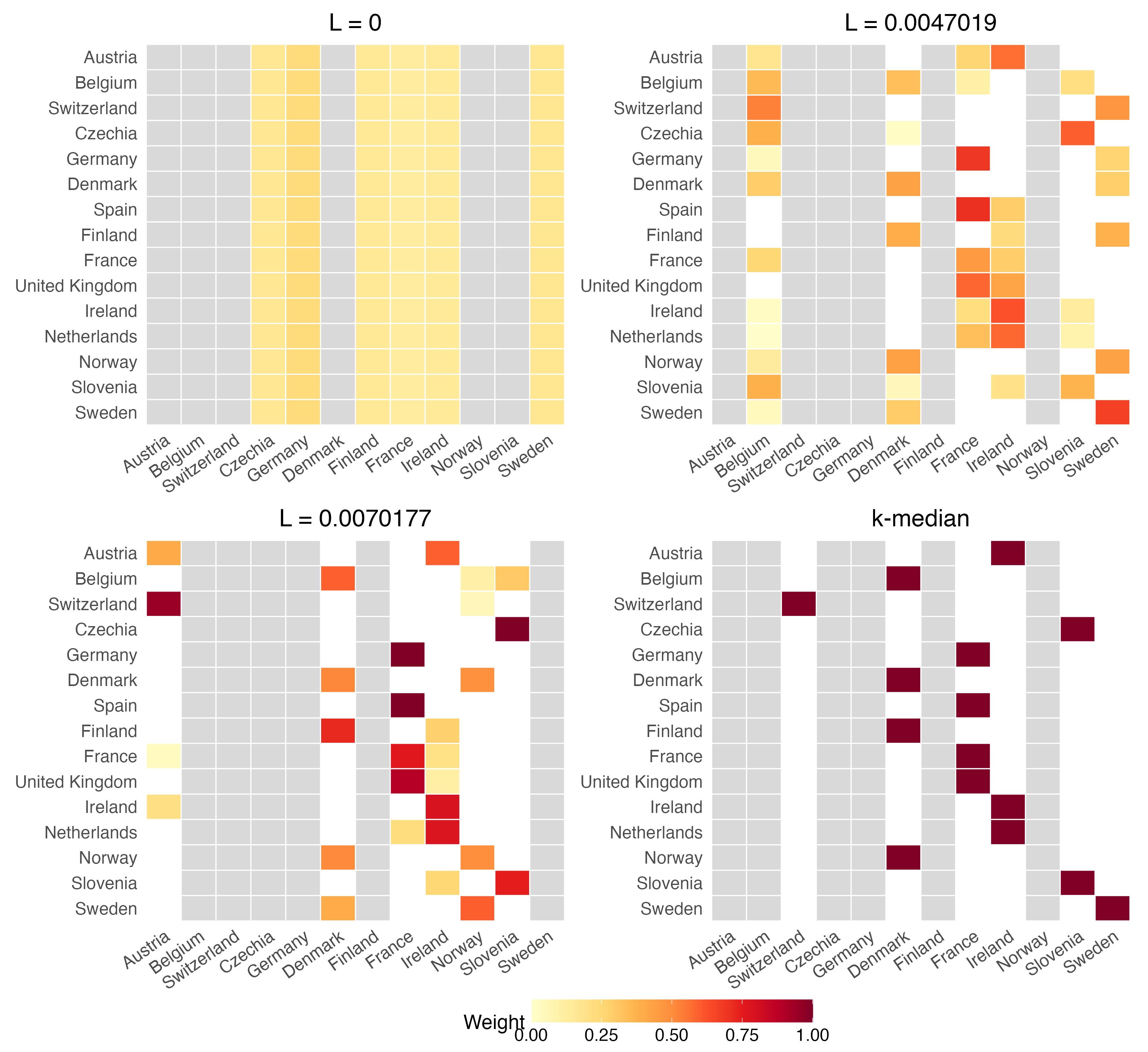}
    \caption{Selected sites by optimizing the RHS of \eqref{eq:site-by-site-upper-bound}.}
    \label{fig:small.c}

\end{figure}
Figure \ref{fig:small.c} displays the site selection results for selected values of $L$ and also the $k$-median solution for reference. In each panel of Figure \ref{fig:small.c}, columns correspond to experimental sites and  rows to policy-relevant sites. To interpret the figure, take the top left panel ($L=0$) as an example. The selected sites are \{Czechia, Finland, France, Germany, Ireland, Sweden\}, and each policy-relevant country is connected with all these selected sites; the darker the highlighted color, the larger the weight. This makes sense because $L=0$ corresponds to a point identified situation in which every selected site matters and the algorithm chooses the  $k$ sites in $\mathcal{S}_E$ that lead to the smallest-variance estimator of the true treatment effect. Indeed, among the selected sites, the higher the variance, the lower the weight. As the magnitude of $L$ increases, each country is connected with fewer sites and we eventually approximate the $k$-median solution that only connects each site with its nearest neighbor.

\subsubsection{Mixed-Integer Program Representation}\label{sec:mixed.integer}

The brute-force approach is not likely to scale well for cases when the total number of sites is large,  or when  the number of selected sites is large. However,
suppose the following condition holds:
\begin{equation}
C_{N\mathscr{_{S}}(s),s}\leq\sqrt{\frac{\pi}{2}}\sigma_{N\mathscr{_{S}}(s)},\forall s\in\mathcal{S}_{P},\forall\mathscr{S}\in\mathcal{A}(k),\label{eq:small}
\end{equation}
implying that the treatment heterogeneity for all policy-relevant sites is
moderate. Furthermore, suppose each experimental site $i\in\mathcal{S}_{E}$, if selected,
has known signal variance $\sigma_{i}^{2}>0$. Then, minimizing the RHS of \eqref{eq:site-by-site-upper-bound} has a  mixed-integer
program representation as follows: 

\begin{align*}
\min_{\{y_i,x_{i,j},m_j^*\}_{i\in\mathcal{S}_E,j\in\mathcal{S}_P}}
& \sum_{j\in\mathcal{S}_P} m_j^*
\Phi\left(
-\sqrt{
\sum_{i\in\mathcal{S}_E}
x_{i,j}
\frac{\left(\max\{m_j^*-C_{i,j},0\}\right)^2}{\sigma_i^2}
}
\right) \\
\text{s.t.}\quad
& y_i\in\{0,1\},\quad x_{i,j}\in\{0,1\},
\qquad \forall i\in\mathcal{S}_E,\ \forall j\in\mathcal{S}_P, \label{constr:nn}\\
& \sum_{i\in\mathcal{S}_E} y_i\le k,\\
& 0\le x_{i,j}\le y_i,
\qquad \forall i\in\mathcal{S}_E,\ \forall j\in\mathcal{S}_P,\\
& \sum_{i \in \mathcal{S}_E} x_{i,j} \geq 1, \;\;\;\;\; \forall j \in \mathcal{S}_P,\\ 
& x_{i,j}(m_j^*-C_{i,j})\geq 0,
\qquad \forall i\in\mathcal{S}_E,\ \forall j\in\mathcal{S}_P,\\
& \frac{
\Phi\left(
-\sqrt{
\sum_{i\in\mathcal{S}_E}
x_{i,j}
\frac{\left(\max\{m_j^*-C_{i,j},0\}\right)^2}{\sigma_i^2}
}
\right)}
{
\phi\left(
-\sqrt{
\sum_{i\in\mathcal{S}_E}
x_{i,j}
\frac{\left(\max\{m_j^*-C_{i,j},0\}\right)^2}{\sigma_i^2}
}
\right)}
\nonumber\\
& =
m_j^*
\frac{
\sum_{i\in\mathcal{S}_E}
x_{i,j}
\frac{\max\{m_j^*-C_{i,j},0\}}{\sigma_i^2}
}{
\sqrt{
\sum_{i\in\mathcal{S}_E}
x_{i,j}
\frac{\left(\max\{m_j^*-C_{i,j},0\}\right)^2}{\sigma_i^2}
}
},
\qquad \forall j\in\mathcal{S}_P.
\end{align*}

The mixed nature of the program comes from the fact that we must select sites (this is the integer component) but also the weights that each policy-relevant site will attach to the experimental sites (this is the continuous variable in the ``mixed'' part). It is well-known that MINLP is an NP-hard problem. Moreover, there are different ways to implement the program. We tried different implementations and we found  using Bonmin with a Python wrapper works well. The computational time of  brute-force and  MINLP is presented in Figure \ref{fig:minlp}. However, we are also cautious because, unlike a global
solver such as Gurobi, Bonmin offers no guarantee of global optimality on a nonconvex problem. So while this particular approach with Bonmin works well, it remains to be investigated how it performs for other datasets. 

\begin{figure}[h!]
    \centering    \includegraphics[width=0.8\linewidth]{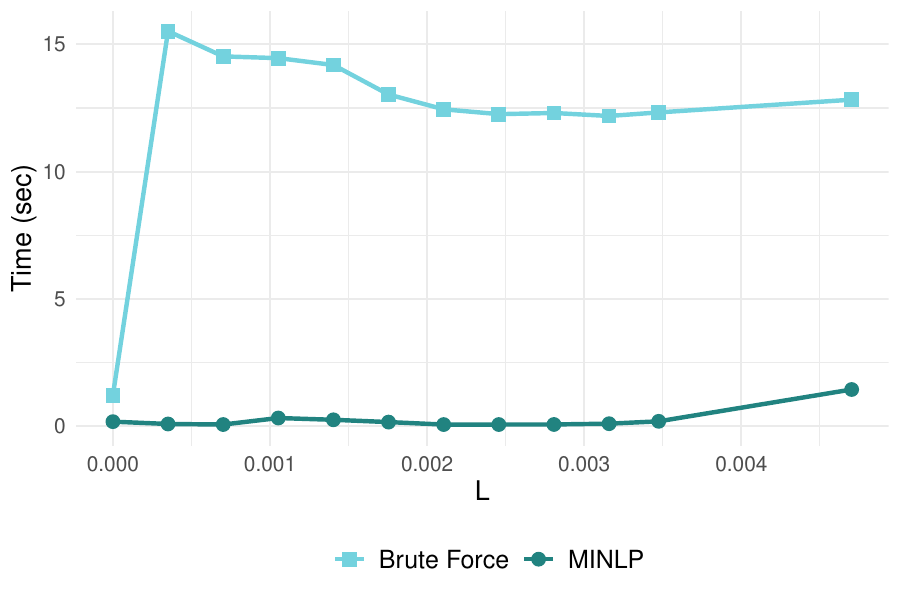}
    \caption{Computational time of brute-force v.s. MINLP}
    \label{fig:minlp}
\end{figure}

\end{document}